\documentclass[11pt,notitlepage]{article}
\usepackage{nopageno}
\usepackage[top=1in,left=1in,right=1in,bottom=1in]{geometry}
\usepackage[T1]{fontenc}
\usepackage[affil-it]{authblk} 
\usepackage{rotating}
\usepackage{amsthm}
\newtheorem{theorem}{Theorem}

\usepackage{graphicx}
\usepackage{balance} 
\usepackage{url}
\usepackage[noend]{algpseudocode}
\usepackage{algorithm}
\usepackage{color,bm,enumitem,subfigure,multirow}

\newcommand{\p}{\mathbf p}
\newcommand{\rv}{\mathbf r}

\newcommand{\vol}{\mathsf{vol}}
\newcommand{\polylog}{\mathsf{polylog}}
\newcommand{\rank}{\mathsf{rank}}
\newcommand{\myparagraph}[1]{\medskip \noindent {\bf #1.}}
\newcommand{\defn}[1]{\emph{\textbf{#1}}}

\newcommand{\codevar}[1]{\mbox{#1}}
\newcommand{\algorithmicbreak}{\textbf{break}}

\newcommand{\set}{sparseSet} 
\newcommand{\vset}{vertexSubset} 
\newcommand{\emap}{\textproc{edgeMap}}
\newcommand{\vmap}{\textproc{vertexMap}}

\newcommand{\size}{\textbf{size}}

\newcommand{\fetchAdd}{\mbox{fetchAdd}}

\newcommand{\text}{\mbox}

\begin{document}

\title{Parallel Local Graph Clustering}

\title{Parallel Local Graph Clustering}

\author{Julian Shun\authorcr UC Berkeley \authorcr jshun@eecs.berkeley.edu
\and Farbod Roosta-Khorasani \authorcr International Computer Science Institute \authorcr and UC Berkeley \authorcr farbod@icsi.berkeley.edu
\and
Kimon Fountoulakis
        \authorcr International Computer Science Institute \authorcr and UC Berkeley \authorcr
        kfount@berkeley.edu
\and
Michael W. Mahoney
       \authorcr International Computer Science Institute \authorcr and UC Berkeley \authorcr
        mmahoney@stat.berkeley.edu
}

\date{}

\maketitle

\begin{abstract}
Graph clustering has many important applications in computing, but due
to growing sizes of graphs, even traditionally fast clustering methods
such as spectral partitioning can be computationally expensive for
real-world graphs of interest.  Motivated partly by this, so-called
local algorithms for graph clustering have received significant
interest due to the fact that they can find good clusters in a graph
with work proportional to the size of the cluster rather than that of
the entire graph. This feature has proven to be crucial in making such
graph clustering and many of its downstream applications efficient in
practice. While local clustering algorithms are already faster than
traditional algorithms that touch the entire graph, they are sequential and there is an
opportunity to make them even more efficient via
parallelization. 
In this paper, we show how to parallelize many of these
algorithms in the shared-memory multicore setting, and we analyze the
parallel complexity of these algorithms.  We present comprehensive
experiments on large-scale graphs showing that our parallel algorithms
achieve good parallel speedups on a modern multicore machine, thus
significantly speeding up the analysis of local graph clusters in the
very large-scale setting.
\end{abstract}

\section{Introduction}\label{sec:intro}
Given a graph, the task of graph clustering is often described as that
of finding a set (or sets) of vertices that are more related, in some
sense, to each other than to other vertices in the graph.
Applications of graph clustering arise in many areas of computing,
including in community detection in social
networks~\cite{Harenberg},
load balancing parallel
computations~\cite{Devine}, unsupervised
learning~\cite{Ng01onspectral}, and optimizing digital map
databases~\cite{Huang1996}.

There are a large number of algorithms for graph clustering, each with
different computational costs and producing clusters with different
properties (see~\cite{Schaeffer2007} for a survey of graph clustering
algorithms). However, most traditional algorithms for graph clustering
require touching the
entire graph at least once, and often many more times. With the
massive graphs that are available today, e.g., extremely large graphs
arising in social media, scientific, and intelligence applications,
these traditionally-fast algorithms can be very computationally 
expensive. A standard example of this can be found in the large-scale
empirical analysis of Leskovec et
al.~\cite{LeskovecLDM09,Leskovec2010} and Jeub et al.~\cite{Jeub15}.
Therefore, there has been a surge of interest in \emph{local graph
  clustering algorithms}, or algorithms whose running time depends
only on the size of the cluster found and is independent of or depends
at most polylogarithmically on the size of the entire graph (we refer to this as a \emph{local running time}).

Local graph clustering was first used by Spielman and
Teng~\cite{Spielman2004} to develop nearly linear-time algorithms for
computing approximately-balanced graph partitions and solving sparse
linear systems. Since then, there have been many improved local graph
clustering algorithms
developed~\cite{Andersen2006,Andersen2009,Gharan2012,Kwok2012}, which
we review in Section~\ref{sec:alg}.  Local graph clustering algorithms
have been used in many real-world applications. For example, Andersen
and Lang~\cite{AndersenL06} use a variant of the algorithm of Spielman
and Teng~\cite{Spielman2004} to identify communities in
networks. Leskovec et al.~\cite{LeskovecLDM09,Leskovec2010} and Jeub
et al.~\cite{Jeub15} use the algorithm of Andersen et
al.~\cite{Andersen2006} as well as other graph clustering algorithms
to study the properties of clusters of different sizes in social and
Web graphs.  A major conclusion
of~\cite{LeskovecLDM09,Leskovec2010,Jeub15} was that large social and
information networks typically have good small clusters as opposed to
large clusters, thus indicating that local algorithms are useful not
only for computational efficiency, but also for identifying clusters
that are more meaningful or useful in practice.  Mahoney et
al.~\cite{MahoneyOV12} and Maji et al.~\cite{MajiVM11} use local
algorithms to obtain cuts for image segmentation and community
detection. These algorithms have also been applied to find communities
in protein networks~\cite{Voevodski2009,Liao09}. 
There have been many other
papers applying local algorithms to community detection,
e.g.,~\cite{Andersen2012,Wu2015,Kloster2014,Kloumann2014,Whang2013,Yang2013,Gleich2012}.

Existing clustering algorithms with local running times are described for the sequential setting,
which is not surprising
since meaningful local clusters in the small to medium-sized graphs
studied in the past tend to be very
small~\cite{LeskovecLDM09,Jeub15}, and hence algorithms to find these
local clusters terminate very quickly. 
Even for these small to
medium-sized graphs, these local algorithms have proven to be very
useful; and currently the applicability of these methods to extremely
large graphs, e.g., those with one billion or more vertices or edges, is
limited by large-scale implementations.  Moreover, with
the massive graphs that are available today, one would like to test
the hypothesis that meaningful local clusters can be larger, and this
will lead to increased running times of local clustering
algorithms. The efficiency of these algorithms can be improved via parallelization.

A
straightforward way to use parallelism is to run many local graph
computations independently in parallel, and this can be useful for
certain applications. However, since all of the local algorithms have
many input parameters that affect both the cluster quality and
computation time, it may be hard to know a priori how to set the input
parameters for the multiple independent computations. We believe that
these local algorithms are more useful in an \emph{interactive setting},
where a data analyst wants to quickly explore the properties of local
clusters found in a graph. In such a setting, an analyst would run a
computation, study the result, and based on that determine what
computation to run next. Furthermore, the analyst may want to
repeatedly remove local clusters from a graph for his or her
application.  To keep response times low, it is important that a
single local computation be made efficient. 
If each run of the algorithm returns nearly
instantaneously rather than in tens of seconds to minutes, this 
drastically improves user experience as well as productivity.
The goal of this paper is
to achieve this via parallelism.

\emph{This paper develops parallel versions of several local
  graph clustering algorithms}---the Nibble algorithm of Spielman and
Teng~\cite{Spielman2004,SpielmanT13}, the PageRank-Nibble algorithm of Andersen et
al.~\cite{Andersen2006}, the deterministic heat kernel PageRank
algorithm of Kloster and Gleich~\cite{Kloster2014}, and the randomized
heat kernel PageRank algorithm of Chung and
Simpson~\cite{Chung2015}. These algorithms all diffuse probability
mass from a seed vertex, and return an approximate PageRank (probability) vector. 
The
vector returned at the end is then processed by a sweep cut procedure to generate a graph
partition (the sweep cut sorts the vertices in non-increasing order of degree-weighted
probability and returns the best partition among all prefixes
of the ordering). 
The approach that we take to parallelizing the diffusion process of these algorithms is to iteratively
process subsets of vertices and their edges until a termination criteria is met. Each iteration processes a possibly different subset of vertices/edges (determined from the previous iteration) in parallel.
We develop an efficient parallel algorithm for
performing the sweep cut as well.  
All of our parallel algorithms return
clusters with the same quality guarantees as their sequential
counterparts. In addition, we prove theoretical bounds on the
computational complexity of the parallel algorithms, showing that
their asymptotic work matches those of the corresponding sequential
algorithms (and thus have local running times) and that most of them have good parallelism.
The only other work on parallelizing local graph clustering algorithms
that we are aware of is a parallelization of the PageRank-Nibble
algorithm of Andersen et al.~\cite{Andersen2006} in the distributed setting by Perozzi et
al.~\cite{Perozzi2014}. However, their algorithm does not have a local
running time since it does work proportional to at least the number of vertices in the graph.

We implement all of our parallel algorithms in the Ligra graph
processing framework for shared-memory~\cite{ShunB2013}. Ligra is
well-suited for these applications because it only does work
proportional to the number of active vertices (and their edges) in
each iteration, which enables local implementations of the
algorithms. In contrast, many other systems (e.g.,
GraphLab~\cite{Powergraph} and Pregel~\cite{Pregel}) require touching
all of the vertices in the graph on each iteration, which would lead to inefficient
implementations of local algorithms.  Our implementations are all lock-free, and use only
basic parallel primitives such as prefix sums, filter, and sorting, as
well as data-parallel functions in Ligra that map computations over
subsets of vertices or edges.

We present a comprehensive experimental evaluation of our parallel
algorithms, and compare them to their sequential counterparts.  Our
experiments on a modern 40-core machine show that our parallel
implementations of the four diffusion methods achieve good speedups
with increasing core count.  For PageRank-Nibble, we describe an
optimization that speeds up both the sequential and the parallel
algorithms.  Our parallel sweep cut procedure also
performs very well in practice, achieving 23--28x speedup on 40 cores.
Due to the efficiency of our algorithms, we are able to
generate network community profile plots (a concept introduced
in~\cite{LeskovecLDM09} and used in~\cite{Jeub15} that quantifies the
best cluster as a function of cluster size) for some of the
largest publicly-available real-world graphs.

\section{Preliminaries}\label{sec:prelims}
\myparagraph{Graph Notation} We denote a graph by $G(V, E)$, where $V$
is the set of vertices and $E$ is the set of edges in the graph. All
of the graphs that we study in this paper are \emph{undirected} and
\emph{unweighted}.
The number of vertices in a graph is $n=|V|$, and the number of
undirected edges is $m=|E|$.  The vertices are assumed to be indexed
from $0$ to $n-1$.  We define $d(v)$ to be the degree of a vertex $v$
(i.e., the number of edges incident on $v$). We define the
\defn{volume} of a set of vertices $S$ to be $\vol(S) = \sum_{v\in
  S}d(v)$, and the \defn{boundary} of $S$ to be $\partial(S) = \{
(x,y) \in E ~\vert ~ x \in S, y \not\in S\}$ (the number of edges
leaving a set).  The \defn{conductance} of a cluster $S$ in a graph is
defined to be $\phi(S) = \big| \partial(S) \big|
/\min(\vol(S),2m-\vol(S))$. This is a widely-used metric to measure
cluster quality. Intuitively, low-conductance vertex sets tend to
correspond to higher-quality clusters because these sets are larger
and have fewer edges to vertices outside of the
set. Figure~\ref{fig:conductance} shows an example graph and the
conductance of several clusters in the graph.

\begin{figure}[!t]
\centering
\includegraphics[width=0.6\columnwidth]{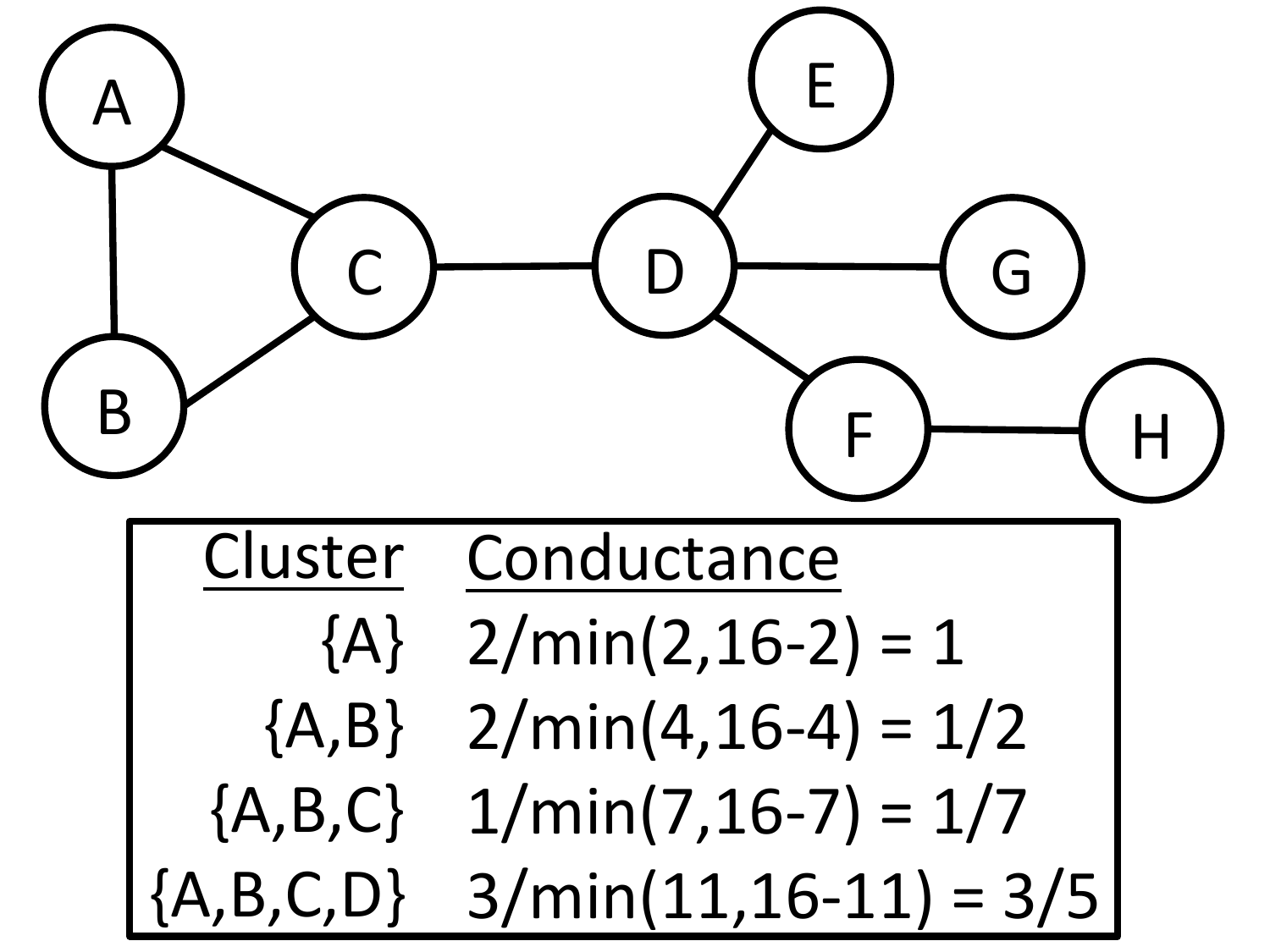}
\caption{An example graph with $n=8$ and $m=8$. The conductance of several clusters are shown.}
\label{fig:conductance}
\end{figure}

\myparagraph{Local Algorithms} \defn{Local graph clustering
  algorithms} have the guarantee that if there exists a cluster $S$
with conductance $\phi$, and one picks a starting vertex in $S$ then
the algorithm returns a cluster of conductance $f(\phi,n)$ with
constant probability.\footnote{$f(\phi,n)$ is a function of $\phi$ and
  $n$ that determines the approximation guarantee of the algorithm. We
  will plug in specific values when we discuss the individual
  algorithms.} They also have a work bound (number of operations) that
depends linearly on the size of the cluster and at most
polylogarithmically on the graph size.
Local algorithms have been used in theory as subroutines to obtain
nearly linear time graph partitioning
algorithms~\cite{Andersen2006,Andersen2007,Andersen2009,SpielmanT13}, which in turn
have applications in solving sparse linear
systems~\cite{Spielman2004}.

\myparagraph{Atomic Operations} A \defn{compare-and-swap} is an atomic
instruction supported on modern multicore machines that takes three
arguments---a memory location, an expected value, and a new value; if the
value stored at the location is equal to the expected value then it atomically
updates the location with the new value and returns true, and
otherwise returns false.  A \defn{fetch-and-add (\fetchAdd{})} takes
two arguments, a location $x$ and a value $y$, and atomically adds $y$
to the value at location $x$. It can be implemented using a loop with
a compare-and-swap until the update is successful.  In the
paper, we use the notation $\mbox{\&}x$ to denote the memory location
of variable $x$.

\myparagraph{Parallel Model}
Algorithms in this paper are analyzed in the work-depth
model~\cite{JaJa92,CLRS}, where \defn{work} is equal to the number of
operations required (equivalent to sequential running time)
and \defn{depth} is equal to the number of time steps required (the
longest chain of sequential dependencies).  The work-depth model is a
natural shared-memory parallel extension of the commonly-used
sequential RAM model~\cite{CLRS} and has been used for decades to
develop parallel algorithms (see, e.g.,~\cite{JaJa92,CLRS}).

By Brent's theorem~\cite{brent1974parallel}, an algorithm with work
$W$ and depth $D$ has overall running time $W/P + D$, where $P$ is the
number of processors available.  Concurrent reads and writes are
allowed in the model, with which a compare-and-swap can be
simulated. A fetch-and-add can be simulated in linear work and
logarithmic depth in the number of updates. A \defn{work-efficient}
parallel algorithm is one whose work asymptotically matches that of
the sequential algorithm, which is important since in practice the
$W/P$ term in the running time often dominates.

\myparagraph{Parallel Primitives} We will use the basic parallel
primitives, prefix sum and filter~\cite{JaJa92}. \defn{Prefix sum}
takes an array $X$ of length $N$, an associative binary operator
$\oplus$ (e.g., the addition operator or the minimum
operator),
and returns the array $(X[0], X[0] \oplus X[1], \ldots,X[0] \oplus
X[1] \oplus \ldots \oplus X[N-1])$.\footnote{This
  definition is for the inclusive version of prefix sum, in contrast to the exclusive
  version of prefix sum, commonly used in parallel algorithms.}
\defn{Filter} takes an array $X$ of length $N$ and a
predicate function $f$, and returns an array $X'$ of length $N' \le N$
containing the elements in $x \in X$ such that $f(a)$ is true, in the
same order that they appear in $X$. Filter can be implemented using
prefix sum, and both require $O(N)$ work and $O(\log N)$
depth~\cite{JaJa92}. We also use parallel comparison sorting, which
for $N$ elements can be done in $O(N\log N)$ work and $O(\log N)$ depth~\cite{JaJa92}, and parallel integer sorting, which can be done in $O(N)$ work and $O(\log N)$ depth with probability $1-1/N^{O(1)}$ (we refer to this as a \emph{high
  probability bound}) for $N$ integers in the range $[1,\ldots,O(N\log^{O(1)}N)]$~\cite{RR89}.

\myparagraph{Sparse Sets} Our implementations use hash tables to
represent a sparse set to store data associated with the vertices
touched in the graph.  This is because we can only afford to do work
proportional to those vertices (and their edges), and cannot
initialize an array of size $|V|$ at the beginning. For sequential
implementations we use the \texttt{unordered\_map} data structure in
STL. For parallel implementations, we use the non-deterministic
concurrent hash table described in~\cite{ShunB14}, which allows for
insertions and searches in parallel. The hash table is a lock-free
table based on linear probing, and makes heavy use of
compare-and-swap and fetch-and-add.  For a batch of $N$ inserts and/or
searches, parallel hashing takes $O(N)$ work and $O(\log N)$ depth
with high probability~\cite{Gil91a}. We set the size of the hash
tables to be proportional to the number of the elements $N$ that we
need to store, so that it can be initialized in $O(N)$ work and
$O(1)$ depth. All of the parallel
algorithms that we present use sparse sets, and so their complexity bounds will
be high probability bounds.

In our pseudocode, we use sparse sets to store key-value pairs where
the key is the vertex ID and the value is the associated data. We use
the notation $\p[k]$ to denote the value in the sparse set $\p$
associated with the key $k$. If we attempt to update data for a
non-existent key $k$ in the sparse set, we assume that prior to
updating, a pair $(k,\bot)$ will be created in the set, where $\bot$
is a zero element defined when creating the set. Both STL's
\texttt{unordered\_map} and the concurrent hash table that we use support this
functionality.
For all of our implementations, $\bot = 0$.

\newcommand{\id}[1]{\emph{\mbox{#1}}}

\myparagraph{Ligra Framework} Our parallel implementations are written
using Ligra, a graph processing framework for shared-memory 
machines~\cite{ShunB2013}. Ligra is very well-suited for implementing
local algorithms since it only does work proportional to the vertices
and edges touched, whereas many other graph processing systems (e.g.,
GraphLab~\cite{Powergraph} and Pregel~\cite{Pregel}) do work
proportional to at least the number of vertices in the graph on every iteration. This feature of
Ligra is crucial in obtaining running times proportional to just the
number of vertices and edges explored. Implementations in Ligra have
been shown to be simple and concise, with performance close to
hand-optimized implementations of the algorithms.
We chose to implement the graph algorithms in shared-memory because
the largest publicly-available real-world graphs can fit in the memory
of a single machine, and shared-memory graph processing has been shown
to be much more efficient than their distributed-memory
counterparts~\cite{ShunB2013,McSherry2015}.

\begin{figure}[!t]
\centering
\includegraphics[width=0.6\columnwidth]{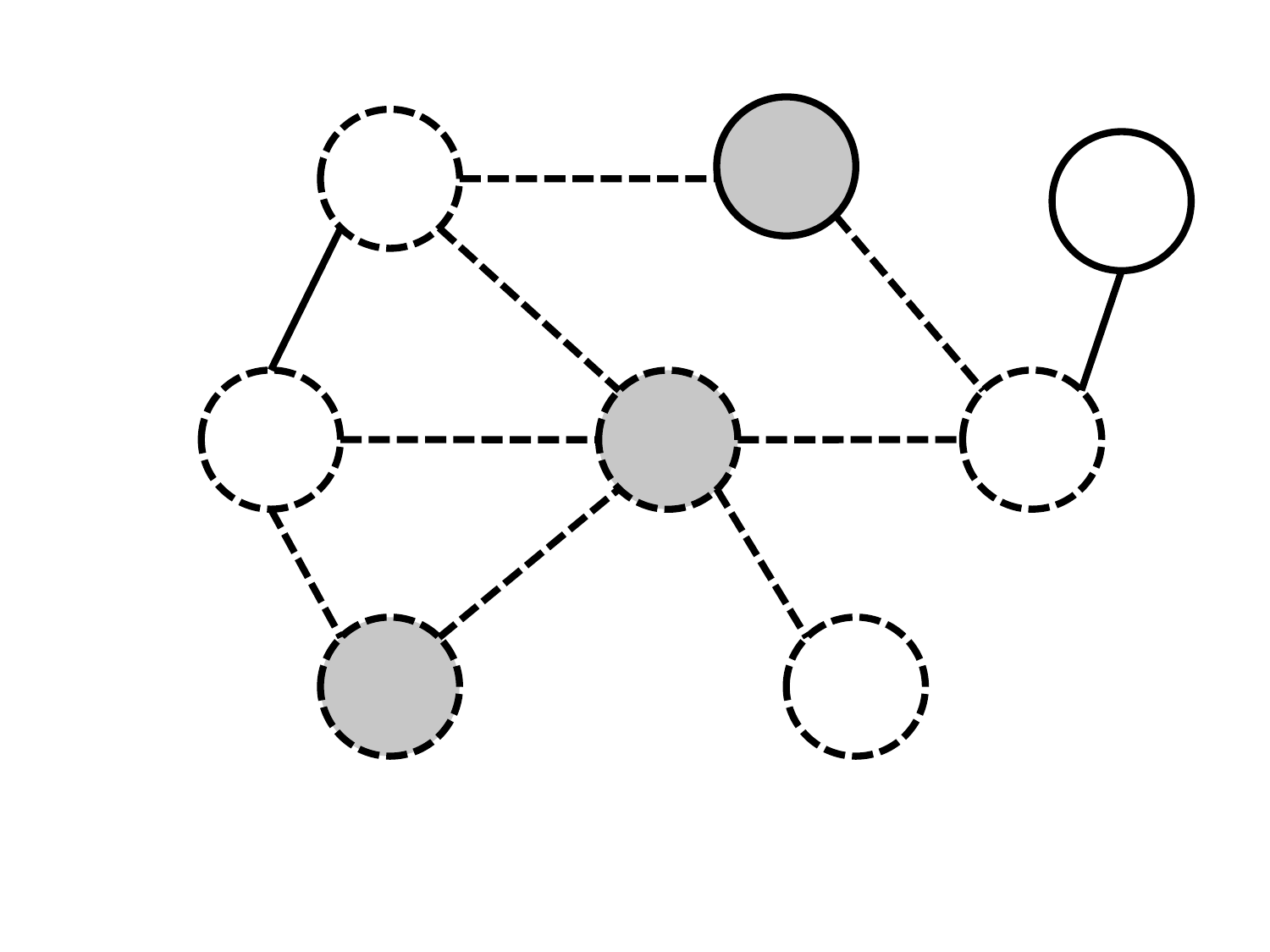}
\caption{An example graph where the shaded vertices are in the \vset{}
  $U$. A \vmap{} applied to $U$ applies a function to data associated
  with the shaded vertices. An \emap{} applied to $U$ applies a
  function to all of the edges incident to $U$ (indicated by dashed
  lines), and can modify the neighbors of $U$ (indicated by dashed
  circles).}
\label{fig:graph}
\end{figure}

Ligra provides a \defn{\vset{}} data structure used for representing a
subset of the vertices, and two simple functions, one for mapping over
vertices and one for mapping over edges. We describe simplified
versions of these functions, which suffices for implementing the
algorithms in this paper (see~\cite{ShunB2013} for the more general
versions). \defn{\vmap{}} takes as input a \vset{} $U$ and a boolean
function $F$, and applies $F$ to all vertices in $U$.
$F$ can side-effect data structures associated with the vertices.
\defn{\emap{}} takes as input a graph $G(V,E)$, \vset{} $U$, and
boolean update function $F$, and applies $F$ to all edges $(u,v)\in
E$ where $u\in U$.
Again, $F$ can side-effect data structures associated with the
vertices.  The programmer ensures the parallel correctness of the
functions passed to \vmap{} and \emap{} by using atomic operations
where necessary. An example graph is shown in Figure~\ref{fig:graph},
where the shaded vertices are in a \vset{}; \vmap{} applies a function
to data associated with the shaded vertices in parallel, and \emap{} applies a
function to the incident edges (dashed lines) and neighbors (dashed
circles) of the shaded vertices in parallel. Note that in some cases multiple
shaded vertices have edges to the same neighbor, so the function must
be correct when run in parallel.

\emap{} is implemented by doing work proportional to the number of vertices in
its input \vset{} and the sum of their outgoing degrees, and processes the vertices
and all of their edges in parallel.  \vmap{} is implemented by doing
work proportional to the number of vertices in its input \vset{}, and
processes all of the vertices in parallel.  Doing work proportional to
only the size of the input \vset{} and its edges makes Ligra efficient
for implementing local graph algorithms that only need to touch part
of the graph.  The Ligra code compiles with either Cilk
Plus or OpenMP for parallelism. We refer the reader to~\cite{ShunB2013} for more
implementation details.

\section{Parallel Algorithms}\label{sec:paralg}
In this section, we review sequential local clustering algorithms and
show how to parallelize them. We describe our parallel algorithms
without any specific setting of parameters, which in prior literature
are often set to specific values for theoretical
purposes. Additionally, we assume the seed set contains just a single
vertex, although all of the algorithms extend to seed sets with
multiple vertices.  

At a high-level, our clustering algorithms are based on iteratively
processing subsets of vertices and their edges in parallel until a
termination criteria is met. For most of the algorithms, we use the
data-parallel \vmap{} and \emap{} functions in Ligra to process each
subset.  One challenge in parallelizing the algorithms is in
guaranteeing a local running time.  To address this challenge, we 
ensure that each iteration only does work proportional to the size of
the subset of vertices and their edges through a careful
representation of sparse sets as well as formulating the algorithms to
use Ligra's functions, which are local when used appropriately. The
second challenge is in identifying which sets of vertices can be
processed in parallel while guaranteeing work-efficiency and
convergence. This requires additional effort for some of the
algorithms.

All of our clustering algorithms compute a vector $\p$, which is
passed to a sweep cut rounding procedure to generate a
cluster. Thus, we first describe the sweep cut procedure and how to
parallelize it in Section~\ref{sec:parSweep}, and then describe our
main clustering routines in
Sections~\ref{sec:nibble}--\ref{sec:rand-hkpr}.

\subsection{Sweep Cut}\label{sec:parSweep}
Often the solution vector $\p$ obtained by spectral partitioning
algorithms contains real numbers instead of being a binary vector that
represents a partition of the graph. Therefore, spectral algorithms
are combined with a rounding procedure which produces a partition from
$\p$ that guarantees a good worst-case approximation to the
combinatorial minimum conductance problem.

The sweep cut procedure is commonly used, and takes as input
a graph $G$ and a vector $\p$ (represented as a sparse set in a local
implementation). It first takes the vertices $v$ with non-zero values
in $\p$ and sorts them in non-increasing order of $\p[v]/d(v)$. This
gives an ordered set $\{v_1,\ldots,v_N\}$, where $N$ is the number of
non-zeros in $\p$, $\p[v_i] > 0$, and $\p[v_i]/d(v_i) \ge
\p[v_{i+1}]/d(v_{i+1})$ for all $i$. It then computes the conductance
of clusters defined by $S_j = \{v_1,\ldots,v_j\}$ for $1 \le j \le N$
and returns the set with smallest conductance. For example, in the
graph in Figure~\ref{fig:conductance} if the ordered set is
$\{A,B,C,D\}$, then the output set of the sweep cut procedure would be
$\{A,B,C\}$ since it has the lowest conductance among the four sets
considered.

The sequential algorithm for the sweep cut first sorts the vertices,
and then iterates through the vertices $v_i$ in increasing order of
$i$, inserts $v_i$ into a set $S$, maintaining the volume $\vol(S)$
and the number of outgoing edges $\partial(S)$ in each iteration. This
also allows the conductance to be computed in each iteration. The
lowest conductance as well as the iteration number $i^*$ that leads to
the lowest conductance is stored, and the final set returned is
$S_{i^*}$. If $S$ is represented as a sparse set, we can check in
constant work if an endpoint of an edge is in $S$. Thus, $\partial(S)$
can be easily updated in each iteration as follows: for each edge
$(v_i,w) \in E$, if $w\in S$ then decrement $\partial(S)$, and
otherwise increment $\partial(S)$. $\vol(S)$ is easily updated by
incrementing it by $d(v_i)$. The sorting costs $O(N\log N)$ work, and
the subsequent iterations costs $O(\vol(S_N))$ work, giving an overall
work of $O(N\log N + \vol(S_N))$.

We now show that the sweep cut procedure can be parallelized
work-efficiently. The challenging part is in computing the conductance of
all of the sets $S_i$ in parallel without increasing the asymptotic
work. A naive approach would be to form all sets $S_i$ for
$1\le i \le N$, and compute $\partial(S_i)$ and $\vol(S_i)$ for each
one independently. However, this leads to $O(N\log N+\sum_{i=1}^N\vol(S_i)) =
O(N\log N+N\vol(S_N))$ work. The following theorem describes a work-efficient
solution, and we illustrate the algorithm with an example afterward.

\begin{theorem}\label{thm:sweep}
A sweep cut can be implemented in $O(N\log N + \vol(S_N))$
work and $O(\log\vol(S_N))$ depth with high probability.
\end{theorem}
\begin{proof}
The initial sort can be parallelized in $O(N\log N)$ work and $O(\log
N)$ depth~\cite{JaJa92}. We then create a sparse set, called $\rank$,
indexed by the vertex identifiers, storing their rank in the sorted
set $S_N=\{v_1,\ldots,v_N\}$.  We create an array $Z$ of size
$2\vol(S_N)$, and for each vertex $v \in S_N$, we look at all edges
$(v,w)\in E$ and if $\rank[w] > \rank[v]$ (case (a)) we create two
pairs $(1,\rank[v])$ and $(-1,\rank[w])$ in the two positions
corresponding to $(v,w)$ in $Z$, and otherwise (case (b)) we create
two pairs $(0,\rank[v])$ and $(0,\rank[w])$. In either case, ranks can
be looked up in $O(1)$ work and if $w \not\in S_N$, we give it a
default rank of $N+1$.  The offsets into $Z$ can be obtained via a
prefix sums computation over the degrees of vertices in $S_N$ in
sorted order. This also allows us to obtain $\vol(S_i)$ for each $i$.
We then sort $Z$ by increasing order of the second value of the pairs
(the order among pairs with equal second value does not matter). Next,
we apply a prefix sums over the sorted $Z$ with the addition operator
on the first value of the pairs.

Since $Z$ is sorted in increasing order of rank, the
final prefix sums gives the number of crossing edges for each
possible cut.  $\partial(S_i)$ for each possible $S_i$ can be obtained
by looking at entries $Z[j]$ where the second value of the pair $Z[j]$
is $i$ and the second value of the pair $Z[j+1]$ is $i+1$. The first
value of $Z[j]$ stores the number of crossing edges for the set $S_i$.
This is because case (a) corresponds to the fact that $v$ is before
$w$ in the ordering---so if they are not on the same side of the cut,
the edge $(v,w)$ will contribute $1$ to the prefix sums of the cut at
any $u$ where $\rank(v) < \rank(u) < \rank(w)$; and if they are on the
same side of the cut, both the $1$ and the $-1$ entries of the edge
will cancel out in the prefix sum, leading to an overall contribution
of $0$ for that edge. Case (b) corresponds to a duplicate edge, and so
does not contribute to the number of crossing edges.

Since we have the volume of all possible sets $S_i$, we can compute
the conductance of each possible cut. A prefix sums using the minimum
operator over the $N$ conductance values gives the cut with the
lowest conductance.

The prefix sums used in the computation contribute $O(\vol(S_N))$ to
the work, and $O(\log \vol(S_N))$ to the depth. Sorting $Z$ by the
rank of vertices using a parallel integer sort takes $O(\vol(S_N))$
work and $O(\log \vol(S_N))$ depth with high probability
since the maximum rank is $N+1 = O(\vol(S_N))$.
Creating the sparse set $\rank$ takes $O(N)$ work and
$O(\log N)$ depth using a hash table.  Including the cost of the
initial sort gives the bounds of the theorem.
\end{proof}

\newcommand{\zs}{Z_{\scriptsize \mbox{sorted}}}
\newcommand{\zss}{Z_{\scriptsize \mbox{sorted, summed}}}

\myparagraph{Example}
We now illustrate the parallel sweep cut algorithm described in
Theorem~\ref{thm:sweep} with an example. Again, let us consider the
example graph in Figure~\ref{fig:conductance}, and the set $\{A,B,C,D\}$,
which we assume has already been sorted in non-increasing order of
$\p[v]/d(v)$. The sparse set $\rank$ stores the following
mapping: $\rank = [A \to 1, B \to 2, C \to 3, D \to 4]$.
To obtain the volume of each of the four possible sets, we
first create an array with the degree of the vertices ordered by
rank, and then apply a prefix sums over it. In the example, the
array of degrees is $[2,2,3,4]$, and the result of the
prefix sums is $[2,4,7,11]$.
The array $Z$
is of size twice the volume of the set $\{A,B,C,D\}$, which is $22$. The
entries of $Z$ are shown below:

{
\begin{tabular}{l}
$Z = [(1,1),(-1,2),(1,1),(-1,3),$ \\
$(0,2),(0,1),(1,2),(-1,3),$ \\
$(0,3),(0,1),(0,3),(0,2),(1,3),(-1,4),$ \\
$(0,4),(0,3),(1,4),(-1,5),(1,4),(-1,5),(1,4),(-1,5)]$\\
\end{tabular}
}

For clarity, we have placed the pairs for each vertex on separate
rows. As an example, let us consider the entries for vertex $B$, which
are on the second row. The first edge of $B$ is to $A$, and the ranks
of $B$ and $A$ are $2$ and $1$, respectively. Since $\rank(A) <
\rank(B)$, we are in case (b) and create the pairs $(0,2)$ and
$(0,1)$. The second edge of $B$ is to $C$, whose rank is $3$. Since
$\rank(C) > \rank(B)$, we are in case (a) and create the pairs $(1,2)$
and $(-1,3)$. The pairs for the other vertices are constructed in a
similar fashion (note that the rank of vertices not in the input set is $5$).

Next, we sort $Z$ by increasing order of the second value of the
pairs. The sorted array $\zs$ is shown below, where pairs with the
same rank (second entry) are on the same row:

{
\begin{tabular}{l}
$\zs = [(1,1),(1,1),(0,1),(0,1),$\\
$(-1,2),(0,2),(1,2),(0,2),$\\
$(-1,3),(-1,3),(0,3),(0,3),(1,3),(0,3),$\\
$(-1,4),(0,4),(1,4),(1,4),(1,4),$\\
$(-1,5),(-1,5),(-1,5)]$\\
\end{tabular}
}

Next we apply a prefix sum over the first value of the pairs in
$\zs$. The resulting array $\zss$ is shown below:

{
\begin{tabular}{l}
$\zss = [(1,1),(2,1),(2,1),(2,1),$\\
$(1,2),(1,2),(2,2),(2,2),$\\
$(1,3),(0,3),(0,3),(0,3),(1,3),(1,3),$\\
$(0,4),(0,4),(1,4),(2,4),(3,4),$\\
$(2,5),(1,5),(0,5)]$\\
\end{tabular}}

With $\zss$, we can obtain the number of crossing edges for each of the
sets. The number of crossing edges for the set $\{A\}$
is found by looking at the pair in $\zss$ whose second value is
$\rank[A] = 1$ and whose next pair's second value is
$\rank[A]+1=2$. This is the $4$'th entry of the array, $(2,1)$, and
the first entry of the pair is the number of crossing edges, which is
$2$. In general the relevant pair containing the number of crossing
edges is the last pair with a given second value, and these pairs all
appear as the last pair of a row of $\zss$ shown above.
We find the number of crossing edges for the set
$\{A,B\}$ to be $2$, $\{A,B,C\}$ to be $1$, and $\{A,B,C,D\}$ to be $3$.

With the number of crossing edges and the volume of each possible set,
we can compute the conductance of each possible set, and take the one
with the lowest conductance using prefix sums. In the example, the
conductance of each set is shown in Figure~\ref{fig:conductance}, and
the set with the lowest conductance is $\{A,B,C\}$.

\subsection{Nibble}\label{sec:nibble}
Spielman and Teng~\cite{Spielman2004} present the first local algorithm
for graph clustering, called \defn{Nibble}, which gives an
approximation guarantee of 
$f(\phi,n) = O(\phi^{1/3}\log^{2/3}n)$ and requires
  $O(|S|\polylog(n)/\phi^{5/3})$ work.\footnote{We use the notation $\polylog(n)$ to mean $\log^{O(1)}n$.} It was later improved to
$f(\phi,n) = O(\sqrt{\phi\log^{3}n})$ 
with $O(|S|\polylog(n)/\phi^2)$ work~\cite{SpielmanT13}.
Their algorithm is based on computing the
distribution of a random walk starting at the seed vertex, but at each
step truncating small probabilities to zero. The sweep cut procedure is
then applied on the distribution to give a partition. 

The Nibble algorithm~\cite{SpielmanT13} takes as input the maximum
number of iterations $T$, an error parameter $\epsilon$, a target
conductance $\phi$, and a seed vertex $x$.
On each iteration (for up to $T$ iterations) Nibble computes a weight
vector, computes a sweep cut on the vector, and either returns the
cluster if its conductance is below $\phi$ or continues running.  Let
$\mathbf{p_0}$ be the initial weight vector, and $\mathbf{p_i}$ denote
the vector on iteration $i$.  $\mathbf{p_0}$ is initialized with
weight $1$ on the seed vertex and $0$ everywhere else. $\mathbf{p_i}$
represents the weights generated by a lazy random walk with truncation from the seed
vertex after $i$ steps, where on each step $\mathbf{p_i}(v)$ is
truncated to $0$ if $\mathbf{p_i}(v) < d(v)\epsilon$.  Computing
$\mathbf{p_{i+1}}$ from $\mathbf{p_i}$ simply requires truncating all
$p_i(v) < d(v)\epsilon$ to $0$ and for each remaining non-zero entry
$\mathbf{p_i}(v)$, sending half of its mass to $\mathbf{p_{i+1}}(v)$
and the remaining to its neighbors, evenly distributed among them.

In practice, computing the sweep cut on each iteration of the
algorithm is unnecessary if one does not know what target conductance
is desirable for the particular graph. As such, we modify the Nibble
algorithm so that it runs for $T$ iterations, returning
$\mathbf{p_T}$, unless there are no vertices on some iteration $i$
such that $\mathbf{p_i}(v) \ge d(v)\epsilon$, in which case we return
$\mathbf{p_{i-1}}$.

Our implementations use sparse sets to represent the $\p$ vector of
the current and previous iteration (vectors from prior iterations can
be safely discarded), so that the work is local.  For the sequential
implementation, the update procedure simply follows the description
above.  Our parallel implementation uses a \vmap{} and \emap{} on each
iteration to update the $\p$ vector of the next iteration. The
pseudocode is shown in Figure~\ref{alg:nibble}. As discussed in
Section~\ref{sec:prelims}, we assume that the size of a sparse set is
proportional to the number of elements it represents, and furthermore
when accessing a non-existing element $v$, the entry $(v,0)$ will be
automatically created in the set.  Initially, the seed vertex is
placed on the frontier and in the weight vector $\p$ (Lines 8--9). On
each iteration the algorithm clears $\mathbf{p'}$, used to store the
next weight vector, applies a \vmap{} to send half of the mass of each
vertex to itself (Line 12), and then an \emap{} to send the remaining
half of the mass evenly among all neighbors (Line
13).
In Line 14, the frontier is updated to contain just the
vertices that have enough mass (equivalent to the thresholding
operation described in~\cite{SpielmanT13}), and only needs to check
the vertices in the frontier at the beginning of the iteration and
their neighbors.  The algorithm breaks on Line 15 if the frontier is
empty, and otherwise updates $\p$ to be the new vector $\p'$ on
Line 16.

\begin{figure}[!t]
\begin{algorithmic}[1] 
\State $\codevar{\set{}}~\p = \{\}$
\State $\codevar{\set{}}~\mathbf{p'} = \{\}$
\Procedure{UpdateNgh}{$s$, $d$} \Comment{passed to \emap{}}
\State $\fetchAdd(\text{\&}\mathbf{p'}[d],\p [s]/(2d(s)))$
\EndProcedure
\smallskip
\Procedure{UpdateSelf}{$v$} \Comment{passed to \vmap{}}
\State $\mathbf{p'}[v] = \p [v]/2$
\EndProcedure

\smallskip
\Procedure{Nibble}{$G$, $x$, $\epsilon$, $T$}
\State $\p = \{(x,1)\}$
\State $\codevar{\vset{}}~\codevar{Frontier} = \{x\}$ \Comment{seed vertex}
\For {$t=1$ to $T$}
\State $\p' = \{\}$
\State $\vmap{}(\codevar{Frontier}, \textproc{UpdateSelf})$
\State $\emap{}(G, \codevar{Frontier}, \textproc{UpdateNgh})$
\State $\codevar{Frontier} = \{v~\vert ~\p'[v] \ge d(v)\epsilon\}$ \Comment{using filter}
\If{$(\size(\codevar{Frontier}) == 0)$} \algorithmicbreak{}
\Else $~\p = \mathbf{p'}$
\EndIf
\EndFor
\State \algorithmicreturn{} $\p$
\EndProcedure
\end{algorithmic}
\caption{Pseudocode for parallel Nibble.} \label{alg:nibble}
\end{figure}

For concreteness, if the input graph is the graph from
Figure~\ref{fig:graph}, and the shaded vertices represent the
frontier, then the \vmap{} on Line 12 would apply the function
\textproc{UpdateSelf} to the $\p$ and $\p'$ entries of the shaded
vertices, and the \emap{} on Line 13 would apply the function
\textproc{UpdateNgh} to the dashed edges, reading from the $\p$
entries of the shaded vertices and writing to the $\p'$ entries of the
dashed vertices. Line 14 would find a new set of shaded vertices
(\vset{}) to represent the next frontier.

Spielman and Teng~\cite{SpielmanT13} show that each iteration of the
algorithm processes $O(1/\epsilon)$ vertices and edges, and so takes
$O(1/\epsilon)$ work ($O((1/\epsilon)\log(1/\epsilon))$ work if
including the per-iteration sweep cut). We can parallelize the updates
across all vertices and edges in each iteration, and for $N$ updates,
the fetch-and-adds together take $O(\log N)$ depth and $O(N)$ work.
The filter performs $O(1/\epsilon)$ work since we only check the
vertices on the frontier and their neighbors.  This gives a
work-efficient parallelization of each iteration in $O(\log
(1/\epsilon))$ depth, and gives an overall depth of $O(T\log
(1/\epsilon))$.  Including the sweep cut on each round, as done
in~\cite{SpielmanT13}, does not increase the depth as it can be
parallelized in logarithmic depth as shown in Theorem~\ref{thm:sweep}.

\begin{theorem}
The parallel algorithm for Nibble requires $O(T/\epsilon)$ work
($O((T/\epsilon)\log(1/\epsilon))$ work if performing a sweep cut per
iteration) and $O(T\log(1/\epsilon))$ depth with high probability.
\end{theorem}

\subsection{PR-Nibble}\label{sec:pr-nibble}
Andersen, Chung, and Lang~\cite{Andersen2006} give an improved local
clustering algorithm, \defn{PageRank-Nibble (PR-Nibble)}, with $f(\phi,n) =
O(\sqrt{\phi\log n})$ and work $O(|S|\polylog(n)/\phi)$.
Their algorithm generates an approximate PageRank vector based on
repeatedly pushing mass from vertices that have enough
residual. Again a sweep cut is applied to the resulting vector
to give a partition. 

The PR-Nibble algorithm~\cite{Andersen2006} takes as input
an error parameter $\epsilon$, a teleportation parameter $0<\alpha<1$, and a
seed vertex $x$. The algorithm maintains two vectors, $\p$ (the PageRank
vector) and $\rv$ (the residual vector).  At the end of the algorithm,
$\p$ is returned to be used in the sweep cut procedure. Initially,
$\p$ is set to $0$ everywhere (i.e., the sparse set representing it is
empty), and $\rv$ is set to store a mass of $1$ on $x$ and $0$ everywhere
else (i.e., the sparse set contains just the entry $(x,1)$).

On each iteration, a vertex $v$ with $\rv(v) \ge d(v)\epsilon$ is chosen
to perform a \defn{push} operation.  Following the description
in~\cite{Andersen2006}, a push operation on vertex $v$ will perform
the following three steps:
\begin{enumerate}[itemsep=0pt,topsep=3pt]
\item $\p[v] = \p[v] + \alpha \rv[v]$
\item For each $w$ such that $(v,w)\in E$: \par
 \hspace{4ex} $\rv[w] = \rv[w] + (1-\alpha)\rv[v]/(2d(v))$
\item $\rv[v] = (1-\alpha)\rv[v]/2$
\end{enumerate}
The PR-Nibble simply repeatedly applies the push operation on a vertex
until no vertices satisfy the criterion $\rv(v) \ge d(v)\epsilon$, at
which point it returns $\p$ and terminates. The work of the algorithm
has been shown to be $O(1/(\alpha\epsilon))$ in~\cite{Andersen2006}.

Again, our sequential implementation simply follows the above
procedure, and uses sparse sets to represent $\p$ and $\rv$ to obtain
the local work bound. As described in~\cite{Andersen2006}, we use a
queue to store the vertices with $\rv(v) \ge d(v)\epsilon$, and
whenever we apply a push on $v$, we check if any of its neighbors
satisfy the criterion, and if so we add it to the back of the
queue. We repeatedly push from $v$ until it is below the threshold.

\myparagraph{An Optimization} We implement an optimization to speed
up the code in practice.  
In particular, we use a more
aggressive implementation of the push procedure as follows:
\begin{enumerate}[itemsep=0pt,topsep=3pt]
\item $\p[v] = \p[v] + (2\alpha/(1+\alpha)) \rv[v]$
\item For each $w$ such that $(v,w)\in E$: \par
 \hspace{4ex} $\rv[w] = \rv[w] + ((1-\alpha)/(1+\alpha))\rv[v]/d(v)$
\item $\rv[v] = 0$
\end{enumerate}
This rule can be shown to approximate the same linear system as the
original rule~\cite{Andersen2006}, and the solution generated can be
shown to have the same asymptotic conductance guarantees as the
original PR-Nibble algorithm.
The work of this
modified algorithm can be shown to also be $O(1/(\alpha\epsilon))$ by
using the same proof as in Lemma 2 of~\cite{Andersen2006} and
observing that at least $(2\alpha\epsilon/(1+\alpha)) d(v) \ge
\alpha\epsilon d(v)$ mass is pushed from $\rv$ per iteration.

Figure~\ref{fig:acl-opt} shows the normalized running times of the
original PR-Nibble algorithm versus our modified version with the optimized update rule for $\alpha=0.01$ and $\epsilon=10^{-7}$ on various input graphs.  In
the experiment, both versions return clusters with the same
conductance for the same input graph.  We see that the optimized
version always improves the running time, and by a factor of
$1.4$--$6.4$x for the graphs that we experimented with.


\begin{figure}[!t]
\centering
\vspace{-10pt}
\includegraphics[width=0.8\columnwidth]{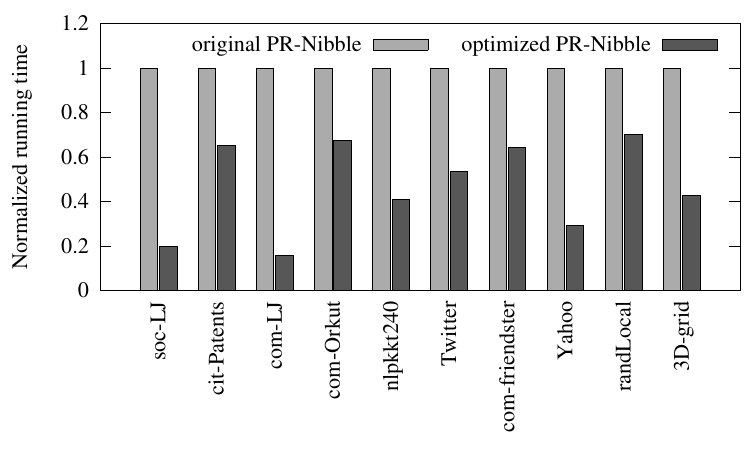}
\vspace{-4pt}
\caption{Running times of the original version of sequential PR-Nibble
  versus the optimized version with $\alpha=0.01$ and
  $\epsilon=10^{-7}$. The running times are normalized to that of
  original PR-Nibble. See Section~\ref{sec:exp} for more information about the graphs and machine specifications.}
\label{fig:acl-opt}
\vspace{-3pt}
\end{figure}


We also tried using a priority queue instead of a
regular queue to store the vertices, where the priority of a vertex
$v$ is the value of $\rv(v)/d(v)$ when it is first inserted into the
queue (with a higher value corresponding to a higher priority). We did
not find this to help much in practice, and sometimes performance was
worse due to the overheads of priority queue operations.


\myparagraph{Parallel Implementation} 
The PR-Nibble algorithm as described above is sequential because each
iteration performs a push on only one vertex. To add parallelism,
instead of selecting a single vertex in an iteration, we select
\emph{all} vertices $v$ where $\rv(v)\ge d(v)\epsilon$, and perform
pushes on them in parallel. This idea was described by Perozzi et
al.~\cite{Perozzi2014}, who implemented it for the distributed setting. Unfortunately their algorithm does not
have a local running time since it does work proportional to at least the
number of vertices in the graph. Here we develop a work-efficient parallel algorithm (and thus with local running time) in the
shared-memory setting based on this idea.

In our parallel algorithm, we
maintain an additional vector $\mathbf{r'}$, which is set to $\rv$ at
the beginning of an iteration, and during the iteration vertices read
values from $\rv$ and write values to $\mathbf{r'}$. At the end of the iteration,
$\rv$ is set to $\mathbf{r'}$. Thus, the pushes use information in the
$\rv$ vector computed from previous iterations, and do not take into
account updates that occur within the current iteration. We also tried
an asynchronous version which only maintains a single $\rv$ vector,
with updates always going to that vector, but found that 
mass would leak when running in parallel due to race conditions, and
it was unclear what the meaning of the solution at the end
was. Developing an asynchronous version that preserves mass and works
well in practice is a direction for future work.

We implement a parallel version of PR-Nibble with the original update
rule (Figure~\ref{alg:pr-nibble}) as well as a version with
the optimized update rule, which requires changing only the update
functions passed to \emap{} (Figure~\ref{alg:pr-nibble2}).  The
implementations are very similar to that of Nibble---the main differences
are in the update rules 
and the fact that PR-Nibble runs until the size of the frontier
becomes empty whereas Nibble will stop after at most $T$ iterations.
 
\begin{figure}[!t]
\vspace{-6pt}
\scriptsize
\begin{algorithmic}[1] 
\State $\codevar{\set{}}~\p = \{\}$
\State $\codevar{\set{}}~\rv = \{\}$
\State $\codevar{\set{}}~\mathbf{r'} = \{\}$
\Procedure{UpdateNgh}{$s$, $d$} \Comment{passed to \emap{}}
\State $\fetchAdd(\text{\&}\mathbf{r'}[d],(1-\alpha)\rv[s]/(2d(s)))$
\EndProcedure
\smallskip
\Procedure{UpdateSelf}{$v$} \Comment{passed to \vmap{}}
\State $\p[v] = \p[v] + \alpha \rv[v]$
\State $\mathbf{r'}[v] = (1-\alpha)\rv[v]/2$
\EndProcedure

\smallskip
\Procedure{PR-Nibble}{$G$, $x$, $\epsilon$, $\alpha$}
\State $\rv = \{(x,1)\}$
\State $\codevar{\vset{}}~\codevar{Frontier} = \{x\}$ \Comment{seed vertex}
\While {$(\size(\codevar{Frontier}) > 0)$}
\State $\vmap{}(\codevar{Frontier}, \textproc{UpdateSelf})$
\State $\emap{}(G, \codevar{Frontier}, \textproc{UpdateNgh})$
\State $\rv = \mathbf{r'}$
\State $\codevar{Frontier} = \{v~\vert ~\rv[v] \ge d(v)\epsilon\}$ \Comment{using filter}
\EndWhile
\State \algorithmicreturn{} $\p$
\EndProcedure
\end{algorithmic}
\caption{Pseudocode for parallel PR-Nibble with the original update rule.} \label{alg:pr-nibble}
\end{figure}

\begin{figure}[!t]
\scriptsize
\begin{algorithmic}[1] 
\Procedure{UpdateNgh}{$s$, $d$} \Comment{passed to \emap{}}
\State $\fetchAdd(\text{\&}\mathbf{r'}[d],((1-\alpha)/(1+\alpha))\rv[s]/d(s))$
\EndProcedure
\smallskip
\Procedure{UpdateSelf}{$v$} \Comment{passed to \vmap{}}
\State $\p[v] = \p[v] + (2\alpha/(1+\alpha))\rv[v]$
\State $\mathbf{r'}[v] = 0$
\EndProcedure
\end{algorithmic}
\caption{Update functions for the optimized version of parallel PR-Nibble. The rest of the code is the same as in Figure~\ref{alg:pr-nibble}.} \label{alg:pr-nibble2}
\vspace{-4pt}
\end{figure}

Unlike Nibble, the amount of work performed in the parallel versions
of PR-Nibble can differ from the sequential version as
the parallel version pushes from all vertices above the threshold with
their residual at the \emph{start of an iteration}.  In particular,
the sequential version selects and pushes a single vertex based
on the most recent value of $\rv$ whereas the parallel version selects
and pushes vertices based on the value of $\rv$ before any of the
vertices in the same iteration have been processed.  The residual of
the vertex when it is pushed in the parallel version can be lower than
when it is pushed in the sequential version, causing less
progress to be made towards termination, and leading to more pushes
overall. However, the following theorem shows that the asymptotic work
complexity of the parallel versions match that of the sequential
versions.

\begin{theorem}
The work of the parallel implementations of PR-Nibble using either
update rule is $O(1/(\alpha\epsilon))$ with high probability.
\end{theorem}
\vspace{-7pt}
\begin{proof}
The proof is very similar to how the work is bounded in the sequential
algorithm~\cite{Andersen2006} and involves showing that the $l_1$ norm
(sum of entries) of the residual vector $\rv$ sufficiently decreases
in each iteration. This idea is also used in~\cite{Perozzi2014}. 

Our proof will be for PR-Nibble using the optimized update rule, but
the proof is similar when using the original rule.  Denote the
residual vector $\rv$ at the beginning of iteration $i$ by $\rv_i$. In
the pseudocode, each iteration will generate $\rv_{i+1}$ by moving
mass from $\rv_i$ into $\mathbf{r'}$. Denote the set of active
vertices in iteration $i$ by $A_i$.  Any vertex $v$ processed in the
iteration will have $\rv_i(v) \ge \epsilon d(v)$ (as the frontier is
defined by vertices that satisfy this property) and do $O(d(v))$ work
to perform the push. It will also contribute $O(d(v))$ work towards
the filter on Line 17.  The work of an iteration is thus $\sum_{v \in
  A_i}O(d(v))$, and the total work of the algorithm is
$\sum_{i=1}^T\sum_{v \in A_i}O(d(v))$ if it runs for $T$ iterations.

A push from vertex $v$ will first set $\mathbf{r'}[v]$ to 0, and then
add a total of $((1-\alpha)/(1+\alpha))\rv_i(v)$ mass to its neighbors' entries
in the $\mathbf{r'}[v]$ vector.  Thus the total contribution to
$|\rv_i|_1 - |\rv_{i+1}|_1$ of vertex $v$ is
$(1-(1-\alpha)/(1+\alpha))\rv_i[v] = (2\alpha/(1+\alpha))\rv_i[v] \ge
(2\alpha/(1+\alpha))\epsilon d(v) \ge \alpha \epsilon d(v)$.  Parallel
updates to $\mathbf{r'}$ will be correctly reflected due to the use of
the atomic fetch-and-add function, and so all vertices $v$ that are on the
frontier in iteration $i$ will contribute at least $\alpha\epsilon
d(v)$ to the difference $|\rv_i|_1 - |\rv_{i+1}|_1$.

We know that $|\rv_1|_1 = 1$ and $|\rv_T|_1 \le 1$,
and so
$1\ge |\rv_1|_1-|\rv_T|_1 = \sum_{i=1}^T(|\rv_i|_1-|\rv_{i+1}|_1)\ge\sum_{i=1}^T\sum_{v\in
  A_i}\alpha\epsilon d(v)$. Rearranging, we have
$\sum_{i=1}^T\sum_{v\in A_i} d(v) \le 1/(\alpha\epsilon)$.  Thus, the total
work of the algorithm is $\sum_{i=1}^T\sum_{v \in A_i}O(d(v)) \le
O(1/(\alpha\epsilon))$. Note that this bound is independent
of $T$.
\end{proof}
\vspace{-7pt}

We also note that both versions of parallel PR-Nibble can be shown to satisfy the same asymptotic
conductance guarantees as the sequential algorithm
of~\cite{Andersen2006}. 

In practice we found that the parallel versions do more work than the
corresponding sequential versions, but benefit from a fewer number of
iterations, each of which can be
parallelized. Table~\ref{table:pr-nibble} shows the number of pushes
for both the sequential and parallel versions of PR-Nibble with the
optimized update rule on several real-world graphs. The table also
shows the number of iterations required for parallel PR-Nibble (the
number of iterations for sequential PR-Nibble is equal to the number
of pushes). We see that the number of pushes of the parallel version
is higher by at most a factor of $1.6$x and usually much less. The
number of iterations is significantly lower than the number of pushes,
indicating that on average there are many pushes to do in an
iteration, so parallelism is abundant.

{
\begin{table}
\small
\centering
\begin{tabular}[!t]{c|r|r|r}
{\bf Input Graph} & Num. Pushes & Num. Pushes & Num. Iterations\\
& (sequential) & (parallel) & (parallel) \\
\hline
soc-LJ & 475,815  & 535,418 & 66 \\
cit-Patents & 4,154,752 & 5,323,710 & 139\\
com-LJ & 763,213 & 917,798 & 77\\
com-Orkut & 516,787  & 803,038 & 166 \\
Twitter & 125,478 & 127,625 & 44\\
com-friendster & 1,731,670 & 2,086,491 & 65\\
Yahoo & 740,410 & 830,457& 96 \\
\end{tabular}
\caption{Number of pushes and number of iterations for PR-Nibble on
  several graphs using $\alpha=0.01$ and $\epsilon=10^{-7}$. More
  information about the graphs can be found in
  Table~\ref{table:inputs} of
  Section~\ref{sec:exp}.}\label{table:pr-nibble}
\end{table}
}

We also implemented a parallel version which in each iteration
processes the top $\beta$-fraction ($0<\beta \le 1$) of the vertices
in the set $\{v~\vert ~\rv[v] \ge d(v)\epsilon\}$ with the highest
$\rv[v]/d(v)$ values. The $\beta$ parameter trades off between additional work
and parallelism.  We found that this optimization helped in practice
for certain graphs, but not by much. Furthermore the best value of
$\beta$ varies among graphs. We do not report the details of the
performance of this variant in this paper due to space constraints.


\subsection{Deterministic Heat Kernel PageRank}\label{sec:hkpr}
Kloster and Gleich~\cite{Kloster2014} present an algorithm for
approximating the heat kernel PageRank distribution of a graph, a
concept that was first introduced in~\cite{Chung2009}.  For a seed
vector $\mathbf s$, random walk matrix $\mathbf P=\mathbf A
\mathbf{D^{-1}}$ ($\mathbf A$ is the adjacency matrix corresponding to
the graph, and $\mathbf D$ is the diagonal matrix with $\mathbf D[i,i]$
containing $d(i)$) and parameters $k$ and $t$, the heat kernel
PageRank vector is defined to be $\mathbf h =
e^{-t}\Large(\sum_{k=0}^\infty\frac{t^k}{k!}(\mathbf
P)^k\Large)\mathbf s$.

The algorithm of Kloster and Gleich, which we refer to as
\defn{HK-PR}, takes as input parameters $N$, $\epsilon$, $t$, and a
seed vertex $x$. It approximates $\mathbf h$ by approximating
$\sum_{k=0}^\infty\frac{t^k}{k!}(\mathbf P)^k$ with its degree-$N$
Taylor polynomial $\sum_{k=0}^N\frac{t^k}{k!}(\mathbf P)^k$, which can
be cast as a linear system.  Their algorithm first computes values
$\psi_k=\sum_{m=0}^{N-k}\frac{k!}{(m+k)!}t^m$ for $k=0,\ldots,N$. It
uses a vector $\rv$, indexed by integer pairs and initialized with
$\rv[(s,0)] = 1$, a vector $\p$ initialized to contain
all $0$'s, and a queue initialized to contain just $(s,0)$. Both $\rv$
and $\p$ are represented using sparse sets. Each iteration
removes an entry $(v,j)$ from the front of the queue, and performs the following
update:
\begin{enumerate}[itemsep=0pt,topsep=3pt]
\item $\p[v] = \p[v] + \rv[(v,j)]$
\item $M = t\cdot \rv[(v,j)]/((j+1)d(v))$
\item for each $w$ such that $(v,w)\in E$: \par
\hspace{2ex} if $j+1==N$: \par\hspace{4ex} then $\p[w] = \p[w] + \rv[(v,j)]/d(v)$ \par
\hspace{2ex}  else: \par
\hspace{4ex}  if $\rv[(w,j+1)] < \frac{e^{-t}\epsilon d(w)}{2N\psi_{j+1}}$ and $\rv[(w,j+1)] + M \ge \frac{e^{-t}\epsilon d(w)}{2N\psi_{j+1}}$: \par\hspace{6ex} then add $(w,j+1)$ to the queue  \par
\hspace{4ex}  $\rv[(w,j+1)] = \rv[(w,j+1)] + M$
\end{enumerate}
The vector $\p$ is output when the queue becomes empty, and a sweep cut
is applied on it to obtain a cluster. The algorithm is deterministic
in that it will generate the same $\p$ every time given the same inputs.

Our sequential implementation follows the procedure above. We observe
that this algorithm can be parallelized by applying the above
procedure on all queue entries $(v,j)$ with the same $j$ value in
order of increasing $j$ because processing these entries only causes
updates to entries $(w,j+1)$ in $\rv$, and can only possibly add
entries with the same form to the queue. Conflicting updates can be
resolved with fetch-and-add. Except for when $j=N-1$, the queue
entries $(v,j)$ only update $\p[v]$, so there will be no conflicting
updates to the $\p$ vector among different vertices $v$. For $j=N-1$,
we can use a fetch-and-add to correctly update the $\p$ vector.

\begin{figure}[!t]
\begin{algorithmic}[1] 
\State $\codevar{\set{}}~\p = \{\}$
\State $\codevar{\set{}}~\rv = \{\}$
\State $\codevar{\set{}}~\mathbf{r'} = \{\}$
\Procedure{UpdateNgh}{$s$, $d$} \Comment{passed to \emap{} on all rounds but the last}
\State $\fetchAdd(\text{\&}\mathbf{r'}[d],t\cdot \rv[s]/((j+1)d(s)))$
\EndProcedure
\smallskip
\Procedure{UpdateNghLast}{$s$, $d$}  \Comment{passed to \emap{} on the last round}
\State $\fetchAdd(\text{\&}\p[d],\rv[s]/d(s))$
\EndProcedure
\smallskip

\Procedure{UpdateSelf}{$v$} \Comment{passed to \vmap{}}
\State $\p[v] = \p[v] + \rv[v]$
\EndProcedure

\smallskip
\Procedure{HK-PR}{$G$, $x$, $N$, $\epsilon$, $t$}
\State precompute $\psi_k=\sum_{m=0}^{N-k}\frac{k!}{(m+k)!}t^m$ for $k=0,\ldots,N$
\State $\rv = \{(x,1)\}$
\State $\codevar{\vset{}}~\codevar{Frontier} = \{x\}$ \Comment{seed vertex}
\State $j=0$
\While {$(\size(\codevar{Frontier}) > 0)$}
\State $\vmap{}(\codevar{Frontier}, \textproc{UpdateSelf})$
\If{$j+1==N$}
\State $\mathbf{r'} = \{\}$
\State $\emap{}(G, \codevar{Frontier}, \textproc{UpdateNgh})$
\State $\rv = \mathbf{r'}$
\State $\codevar{Frontier} = \{v~\vert ~\rv[v] \ge \frac{e^t\epsilon d(v}{2N\psi_{j+1}(t)}\}$ \Comment{using filter}
\State $j = j+1$
\Else \Comment{last round}
\State $\emap{}(G, \codevar{Frontier}, \textproc{UpdateNghLast})$
\State \algorithmicbreak{}
\EndIf
\EndWhile
\State \algorithmicreturn{} $\p$
\EndProcedure
\end{algorithmic}
\caption{Pseudocode for parallel HK-PR.} \label{alg:hkpr}
\end{figure}

The pseudocode for our parallel implementation is shown in
Figure~\ref{alg:hkpr}.  We no longer need to index $\rv$ with the
second integer $j$, since this is now implicitly captured by the
iteration number of the algorithm. We use $\rv$ to store the values
for the current iteration and $\mathbf{r'}$ to store the values for
the next iteration. We initialize $\rv$ and the frontier to contain
just the seed vertex (Lines 12--13). On each iteration we apply a
\vmap{} to update the $\p$ values of the frontier vertices (Line
16). If it is not the last round ($j+1 < N$), we apply an \emap{} to
update the $\mathbf{r'}$ values of the neighbors of the frontier using
the update rule with fetch-and-add (Line 19), and generate a new
frontier for the next round based on the threshold specified in the
sequential algorithm using a filter (Line 21). For the last round, we apply an
\emap{} to update the $\p$ values of the neighbors of the frontier
(Line 24).  This parallel algorithm applies the same updates as the
sequential algorithm and thus the vector returned is the same.

As shown in~\cite{Kloster2014}, the sequential algorithm explores
$O(Ne^t/\epsilon)$ edges, leading to an overall work of
$O(Ne^t/\epsilon)$.  Our parallel algorithm only does a constant
factor more work (the work of the filter is proportional to the number
of edges processed in the iteration), and so also has a work bound of
$O(Ne^t/\epsilon)$.  The depth for iteration $j$ is $O(\log U_j)$ for
the fetch-and-adds and filter, where $U_j$ is the number of vertices
and edges processed in the iteration, and $\sum_{j=0}^NU_j =
O(Ne^t/\epsilon)$. This gives an overall depth of $\sum_{j=0}^NO(\log
U_j) = O(Nt\log (1/\epsilon))$, where we use the fact that the
logarithm is a concave function and the sum is maximized when all
$U_j$'s are equal. The initialization on Line 11 can be done in
$O(N^2)$ work and $O(\log N)$ depth using prefix sums, which is
work-efficient. This gives the following theorem.

\begin{theorem}
The parallel algorithm for HK-PR requires $O(N^2+Ne^t/\epsilon)$ work and
$O(Nt\log(1/\epsilon))$ depth with high probability.
\end{theorem}

As noted in~\cite{Kloster2014}, in practice $N$ is set to at most
$2t\log(1/\epsilon)$, so the $O(N^2)$ term in the work is a lower-order
term.

\subsection{Randomized Heat Kernel PageRank}\label{sec:rand-hkpr}
Chung and Simpson~\cite{Chung2015} describe a randomized algorithm for
approximating the heat kernel PageRank based on running a sample of
random walks and computing the distribution of the last vertices
visited in the random walks. We refer to this algorithm as \defn{rand-HK-PR}.
The algorithm takes as input parameters $x$, $N$, $K$, and $t$, where
$x$ is the seed vertex, $N$ is the number of random walks to perform,
$K$ is the maximum length of a random walk, and $t$ is a parameter to
the heat kernel PageRank equation, as defined in Section~\ref{sec:hkpr}.
Rand-HK-PR runs $N$ random walks starting from $x$, where each step of
the walk visits a neighbor of the current vertex with equal
probability, and the walk length is $k$ with probability
$e^{-t}t^k/k!$. The maximum walk length is set to $K$. It maintains a
vector $\p$, where $\p[v]$ stores the number of random walks that
ended on vertex $v$.
The sequential algorithm stores $\p$ as a sparse set, initialized to
be empty, and executes one random walk at a time for $N$ times, each time
incrementing $\p[v]$ by $1$ where $v$ is the last vertex visited in the random
walk. The vector returned and passed to the sweep cut procedure is
$(1/N)\p$. 

The algorithm is easily parallelizable by running all of the random
walks in parallel and incrementing $\p$ using a
fetch-and-add. However, we found that naively implementing this approach led to poor
speed up since many random walks end up on the same vertex causing
high memory contention when using fetch-and-adds to update that
location. Instead, we keep an array $A$ of length $N$ and have the
$i$'th random walk store its destination vertex into $A[i]$. Afterward
we sort $A$, and compute the number of random walks ending on each
vertex using prefix sums and filter. In particular, we create an
auxiliary array $B$, and for all locations $i$ in the sorted $A$ array
such that $A[i] \neq A[i-1]$, $B[i] = i$, and otherwise $B[i] =
-1$. Filtering out all $-1$ entries in $B$ gives the offsets where the
entries in $A$ differ, and the difference between consecutive offsets
gives the number of random walks ending at a particular value, which 
allows us to compute $\p$. The prefix sums and filter take $O(N)$ work
and $O(\log N)$ depth. To perform the sorting in $O(N)$ work and
$O(\log N)$ depth, we compute a mapping from each last-visited vertex
to an integer in $[0,\ldots,N]$ using a parallel hash table, so that
the maximum value is bounded by $N$, and then use a parallel integer
sort~\cite{RR89} on the mapped values.

The sequential algorithm takes $O(NK)$ work since $N$ random walks of
length $O(K)$ are executed. The parallel algorithm takes $O(NK)$ work and $O(K+\log N)$ depth, as all random walks are run in
parallel and each takes $O(K)$ steps. This gives a work-efficient parallel algorithm, and we have the following theorem: 

\begin{theorem}
The parallel algorithm for rand-HK-PR takes $O(NK)$ work and $O(K+\log
N)$ depth with high probability.
\end{theorem}

In contrast to the previous three
algorithms, we do not need to use the Ligra functions for this
algorithm since the random walks are independent and only process a
single vertex in each iteration.

We note that Chung and Simpson describe a
distributed version of their algorithm~\cite{Chung2015a}. This
work differs from ours in that it assumes the input
graph is the network of processors and studies how to compute local clusters
with limited communication. In contrast, we study the setting where
the input graph is independent of the number of processors or
their layout, and our goal is to speed up the computation by taking
advantage of parallel resources.

\begin{table}
\small
\centering
\begin{tabular}[!t]{c|r|r}
{\bf Input Graph} & Num. Vertices & Num. Edges$^\dag$\\
\hline\hline
soc-LJ & 4,847,571 & 42,851,237 \\
cit-Patents & 6,009,555 & 16,518,947\\
com-LJ & 4,036,538 & 34,681,189 \\
com-Orkut & 3,072,627 & 117,185,083 \\
nlpkkt240 & 27,993,601 & 373,239,376 \\
Twitter & 41,652,231 & 1,202,513,046 \\
com-friendster & 124,836,180 & 1,806,607,135 \\
Yahoo & 1,413,511,391 & 6,434,561,035 \\
randLocal (synthetic) & 10,000,000 & 49,100,524\\
3D-grid (synthetic) & 9,938,375 & 29,815,125\\

\end{tabular}
\caption{Graph inputs used in experiments. $^\dag$Number of unique
  undirected edges.}\label{table:inputs}
\end{table}

{
{\setlength{\tabcolsep}{2.5pt}
\begin{sidewaystable}[!ph]
\scriptsize
\centering
\begin{tabular}[!t]{c|rr|rr|rr|rr|rr|rr|rr|rr|rr|rr}
Algorithm & \multicolumn{2}{c|}{soc-LJ} & \multicolumn{2}{c|}{cit-Patents} & \multicolumn{2}{c|}{com-LJ} & \multicolumn{2}{c|}{com-Orkut} & \multicolumn{2}{c|}{nlpkkt240} & \multicolumn{2}{c|}{Twitter} & \multicolumn{2}{c|}{com-friendster} & \multicolumn{2}{c|}{Yahoo} & \multicolumn{2}{c|}{randLocal} & \multicolumn{2}{c}{3D-grid}\\

& $T_{1}$ & $T_{40}$ & $T_{1}$ & $T_{40}$ & $T_{1}$ & $T_{40}$ & $T_{1}$ & $T_{40}$ & $T_{1}$ & $T_{40}$ & $T_{1}$ & $T_{40}$ & $T_{1}$ & $T_{40}$ & $T_{1}$ & $T_{40}$ & $T_{1}$ & $T_{40}$ & $T_{1}$ & $T_{40}$ \\
\hline\hline
Parallel Nibble  & 19.20 & 0.664 & 3.06 & 0.195 & 15.30 & 0.566 & 3.83 & 0.290 & 0.02 & 0.047 & 1.60 & 0.167 & 4.53 & 0.297 & 1.88 & 0.214 & 2.47 & 0.196 & 0.01 & 0.040\\
Sequential Nibble    & 24.90 & -- & 4.16 & -- & 17.80 & -- & 8.26 & -- & 0.02 & -- & 7.93 & -- & 8.40 & -- & 1.97 & -- & 5.88 & -- & 0.01 & --\\
\hline
Parallel PR-Nibble   & 6.08 & 0.350 & 21.50 & 0.978 & 7.59 & 0.436 & 9.86 & 0.572 & 1.82 & 0.530 & 2.50 & 0.177 & 8.62 & 0.492 & 6.16 & 0.623 & 28.90 & 1.250 & 1.90 & 0.502\\
Sequential PR-Nibble  & 2.17 & -- & 8.92 & -- & 2.74 & -- & 1.94 & -- & 2.21 & -- & 0.70 & -- & 7.97 & -- & 3.56 & -- & 9.46 & -- & 1.41 & --\\

\hline
Parallel HK-PR    & 56.60 & 1.710 & 11.00 & 0.430 & 39.70 & 1.230 & 28.80 & 0.918 & 0.02 & 0.049 & 7.55 & 0.434 & 20.40 & 0.866 & 6.40 & 0.377 & 11.90 & 0.479 & 0.01 & 0.049\\
Sequential HK-PR     & 241.00 & -- & 24.20 & -- & 156.00 & -- & 93.40 & -- & 0.03 & -- & 28.60 & -- & 68.30 & -- & 11.30 & -- & 32.00 & -- & 0.02 & --\\
\hline
Parallel rand-HK-PR   & 129.00 & 2.080 & 23.10 & 0.634 & 117.00 & 1.610 & 83.90 & 1.390 & 23.70 & 0.568 & 75.70 & 1.390 & 45.30 & 0.825 & 42.40 & 0.901 & 33.90 & 0.744 & 21.20 & 0.496\\
Sequential rand-HK-PR   & 195.00 & -- & 43.20 & -- & 154.00 & -- & 121.00 & -- & 18.90 & -- & 160.00 & -- & 85.60 & -- & 67.10 & -- & 51.70 & -- & 19.10 & --\\
\hline
Parallel Sweep   & 4.37 & 0.189 & 1.98 & 0.075 & 5.10 & 0.217 & 6.73 & 0.285 & 0.01 & 0.003 & 75.00 & 3.180 & 5.00 & 0.221 & 43.10 & 1.550 & 2.24 & 0.083 & 0.01 & 0.01\\
Sequential Sweep & 2.57 & -- & 1.51 & -- & 3.20 & -- & 3.40 & -- & 0.01 & -- & 40.10 & -- & 3.66 & -- & 6.43 & -- & 1.72 & -- & 0.01 & --\\

\end{tabular}
\caption{Running times (seconds) of sequential and parallel local graph clustering algorithms and the sweep cut procedure. $T_1$ is the single-thread time and $T_{40}$ is the parallel time using 40 cores with hyper-threading. The parameters are set as follows: $T=20$ and $\epsilon=10^{-8}$ for Nibble; $\alpha=0.01$ and $\epsilon=10^{-7}$ for PR-Nibble; $t=10$, $N=20$, and $\epsilon=10^{-7}$ for HK-PR; and $t=10$, $K=10$, and $N=10^8$ for rand-HK-PR. The running times for the sweep cut are based on using the output of Nibble.
}
\label{table:numbers}
\end{sidewaystable}
}

\section{Experiments}\label{sec:exp}
We present an experimental study of both
parallel and sequential implementations of the local clustering
algorithms described in Section~\ref{sec:paralg} on large-scale
undirected graphs.  All of our implementations are
available at {\url{https://github.com/jshun/ligra/}}.

\myparagraph{Input Graphs} We use a set of unweighted undirected
real-world and synthetic graphs, whose sizes are shown in Table~\ref{table:inputs}. 
We obtained the
\defn{soc-LJ}, \defn{cit-Patents}, \defn{com-LJ},
\defn{com-Orkut}, and \defn{com-friendster} real-world graphs from {
  \url{http://snap.stanford.edu/}}. \defn{nlpkkt240} is a graph derived from a matrix of a
constrained optimization problem from {
  \url{http://www.cise.ufl.edu/research/sparse/matrices/}}.
\defn{Twitter} is a symmetrized version of a snapshot of the Twitter
network~\cite{Kwak2010}.  
\defn{Yahoo} is a symmetrized version of a Web graph from
{\footnotesize \url{http://webscope.sandbox.yahoo.com}}. 
\defn{randLocal} is a random graph where
every vertex has five edges to neighbors chosen with probability
proportional to the difference in the neighbor's ID value from the
vertex's ID. 
\defn{3D-grid} is a synthetic grid graph in 3-dimensional space where
every vertex has six edges, each connecting it to its 2 neighbors in
each dimension.  We remove all self and duplicate edges from the
graphs.

\myparagraph{Experimental Setup}
We run our experiments on a 40-core Intel machine (with two-way
hyper-threading\footnote{A form of simultaneous multithreading developed by Intel in which there are two logical processors per physical core.}) with $4\times 2.4\mbox{GHz}$ Intel 10-core E7-8870
Xeon processors (with a 1066MHz bus and 30MB L3 cache) and
256\mbox{GB} of main memory. 
The parallel programs are compiled with Cilk Plus from
the \texttt{g++} compiler (version 4.8.0) with the \texttt{-O3}
flag (they can also be compiled with OpenMP with similar performance).
Our parallel implementations are all written in Ligra~\cite{ShunB2013}, with
the exception of rand-HK-PR, which
does not need Ligra's functionality. The parallel implementations of
prefix sum, filter, comparison sort, and integer sort that we
use are from the Problem Based Benchmark Suite~\cite{SBFG}. The
concurrent hash table for representing sparse sets is
from~\cite{ShunB14}. For both sequential and parallel PR-Nibble, we
report performance of the versions using the optimized update rule as
described in Section~\ref{sec:pr-nibble}.

\begin{figure*}[!t]
\centering
\subfigure[Nibble running time]{
\includegraphics[width=0.31\textwidth]{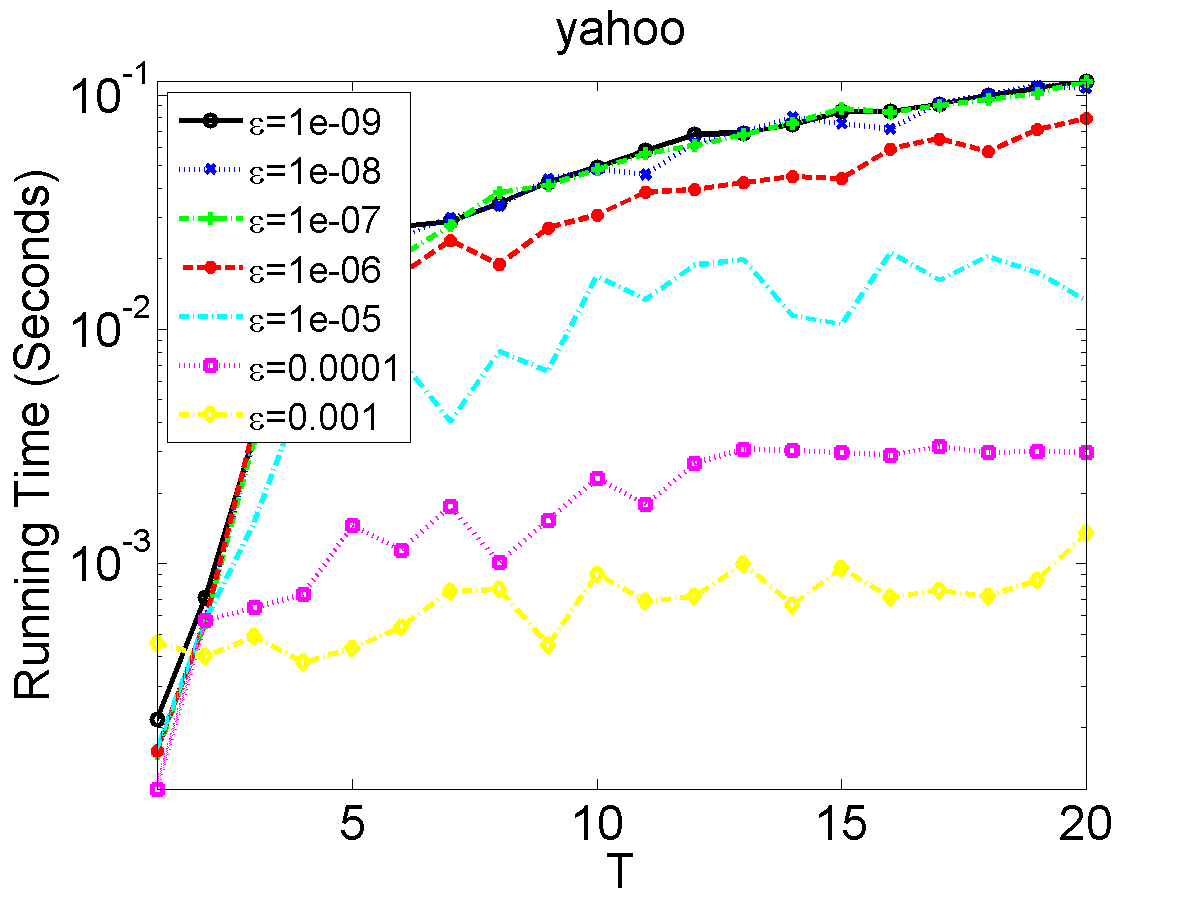}
\label{fig:nibble-time}
}
\subfigure[Nibble conductance]{
\includegraphics[width=0.31\textwidth]{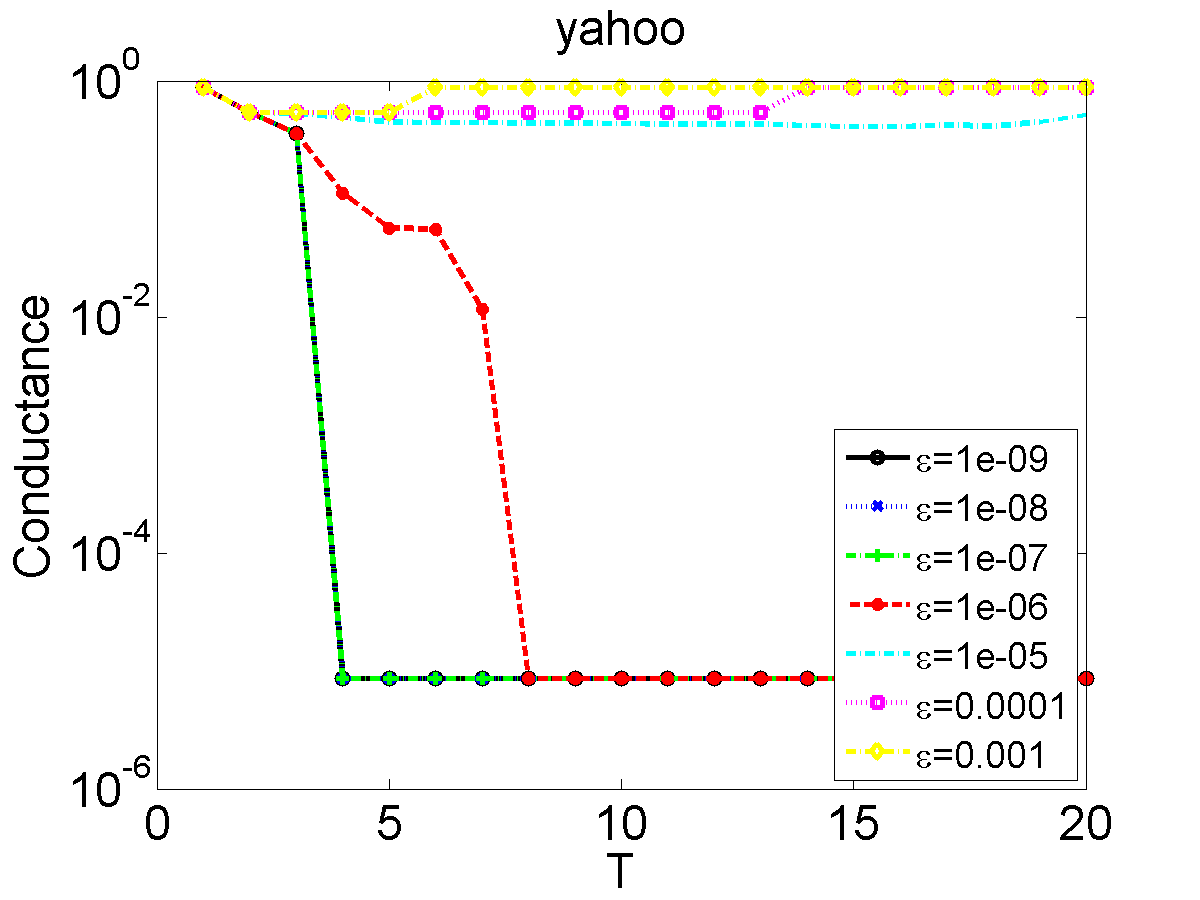}
\label{fig:nibble-conductance}
}
\subfigure[PR-Nibble running time]{
\includegraphics[width=0.31\textwidth]{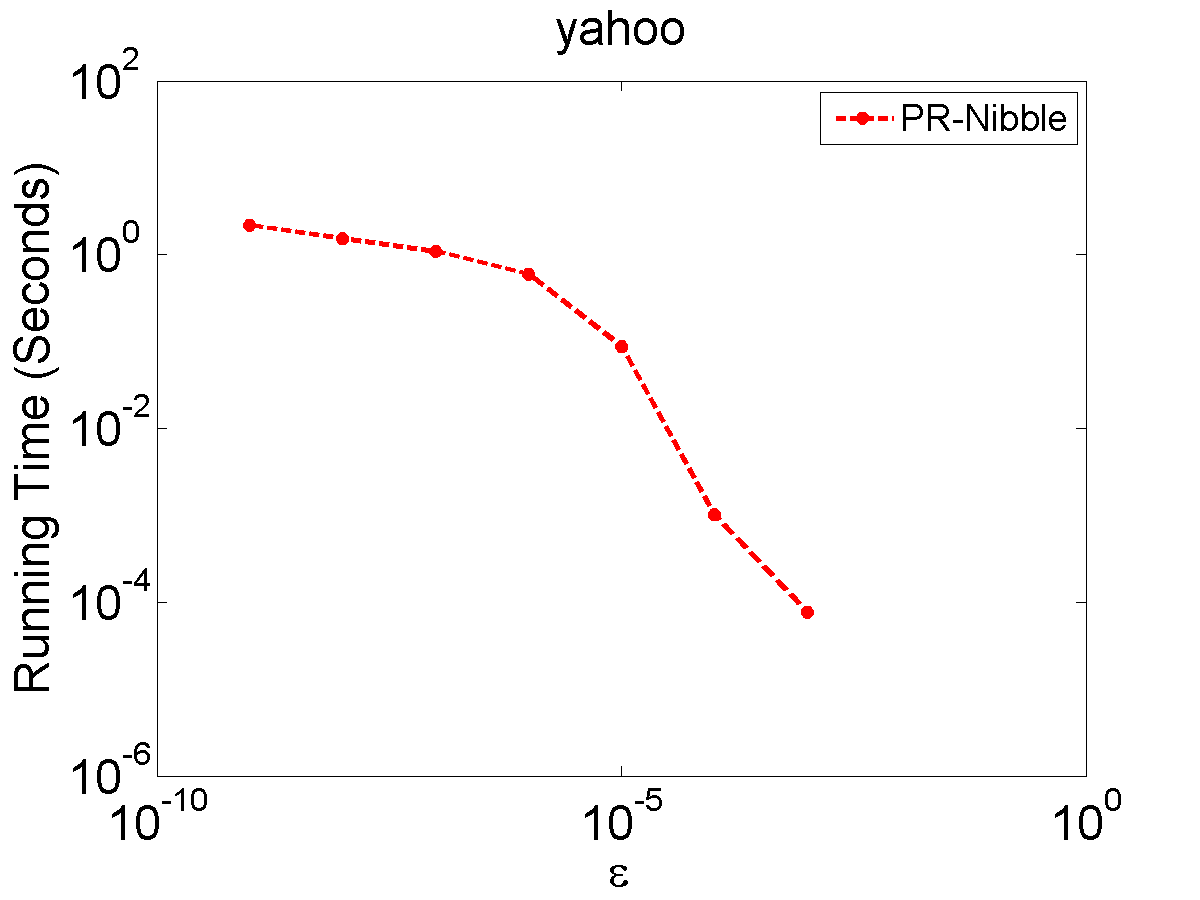}
\label{fig:pr-nibble-time}
}

\subfigure[PR-Nibble conductance]{
\includegraphics[width=0.31\textwidth]{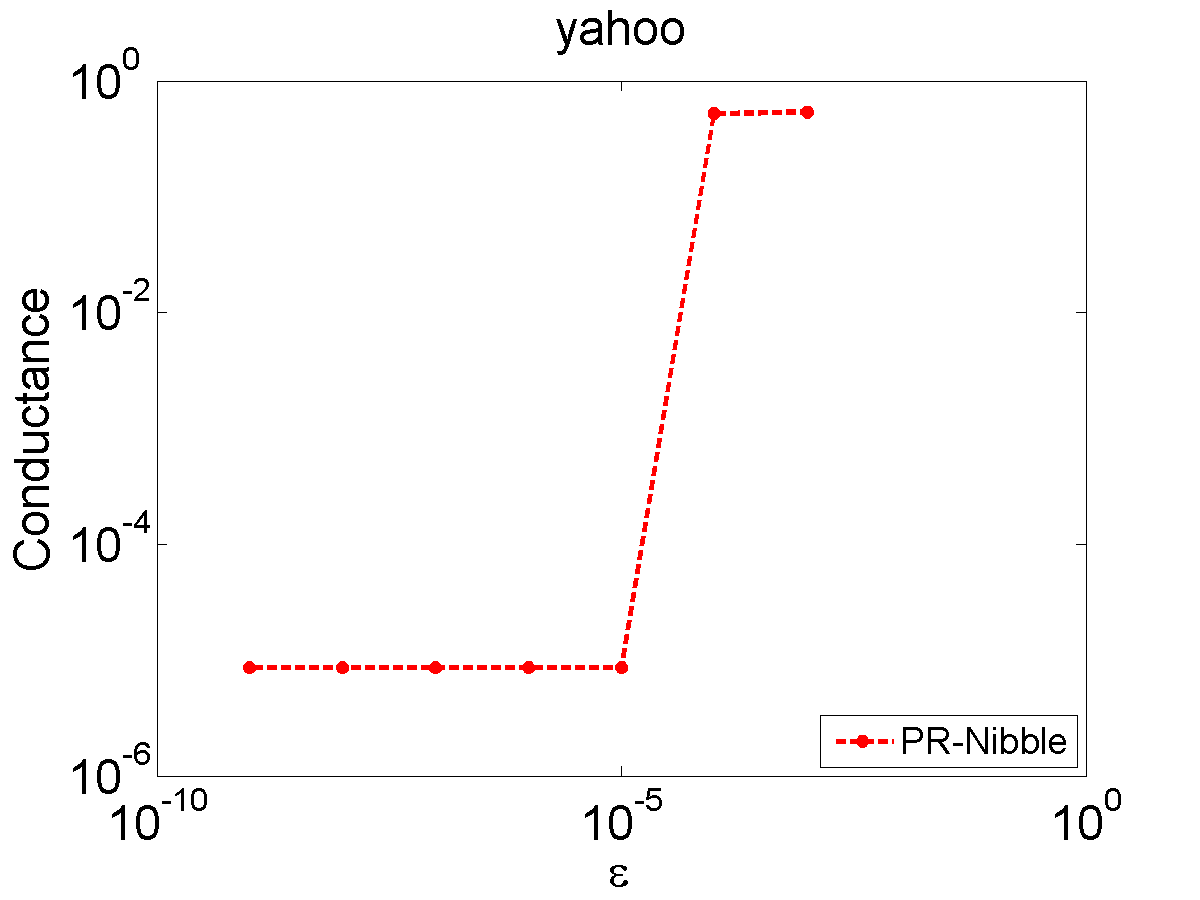}
\label{fig:pr-nibble-conductance}
}
\subfigure[HK-PR running time]{
\includegraphics[width=0.31\textwidth]{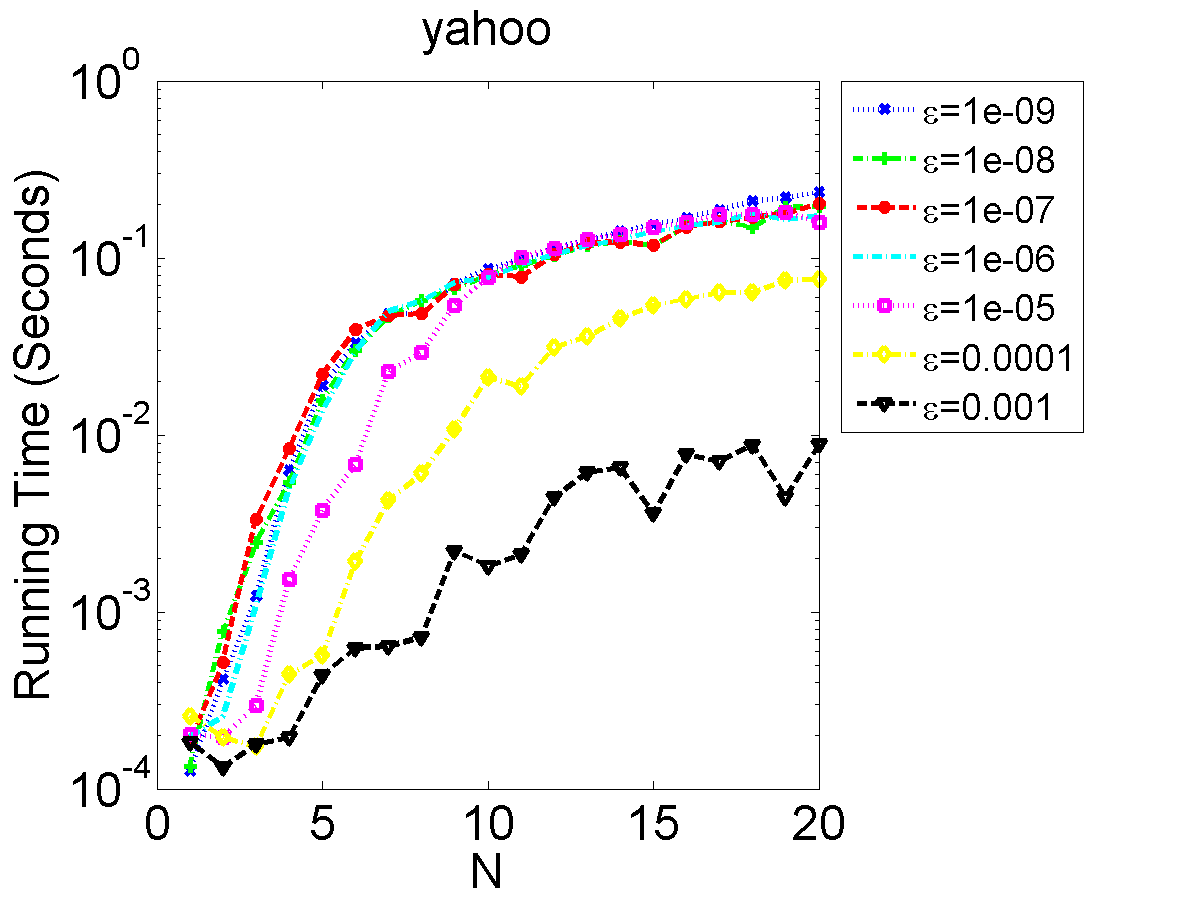}
\label{fig:hk-pr-time}
}
\subfigure[HK-PR conductance]{
\includegraphics[width=0.31\textwidth]{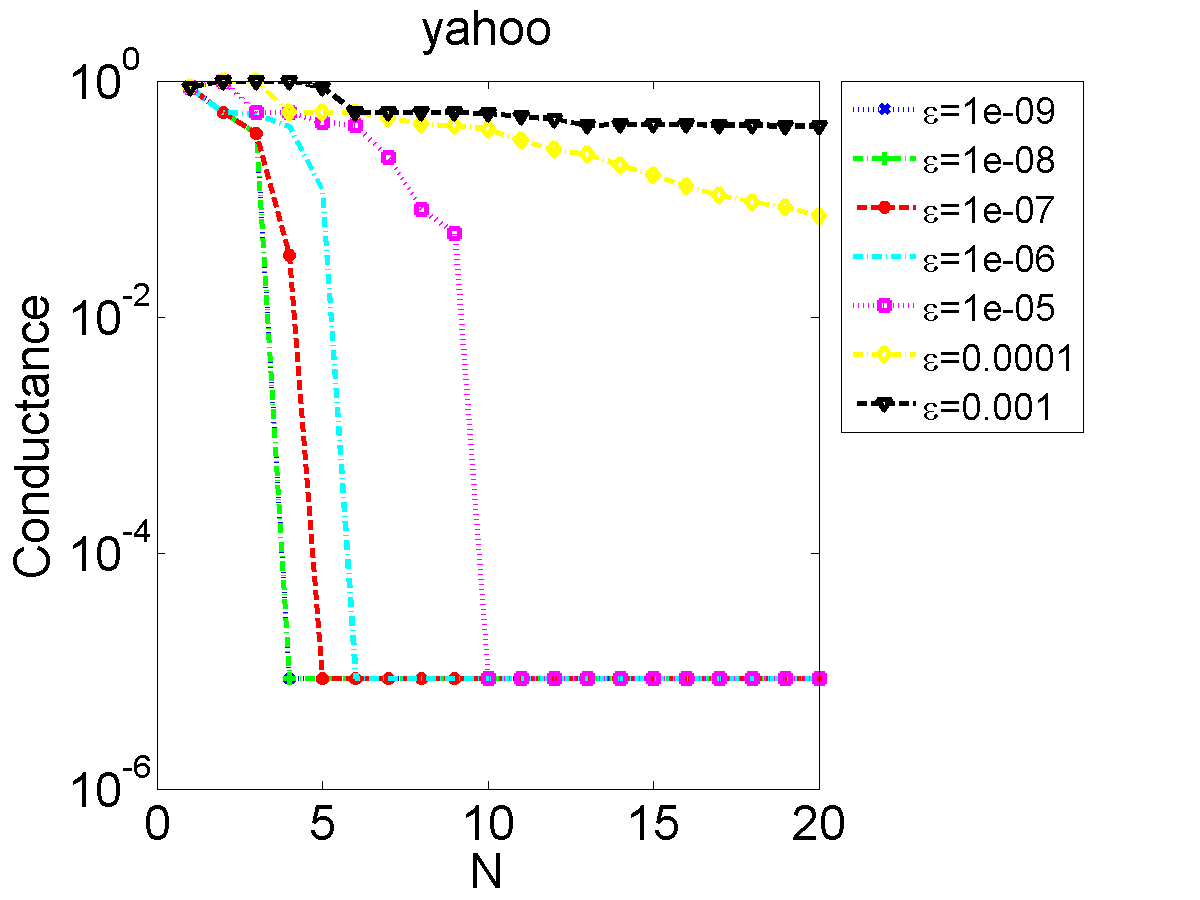}
\label{fig:hk-pr-conductance}
}

\subfigure[rand-HK-PR running time]{
\includegraphics[width=0.31\textwidth]{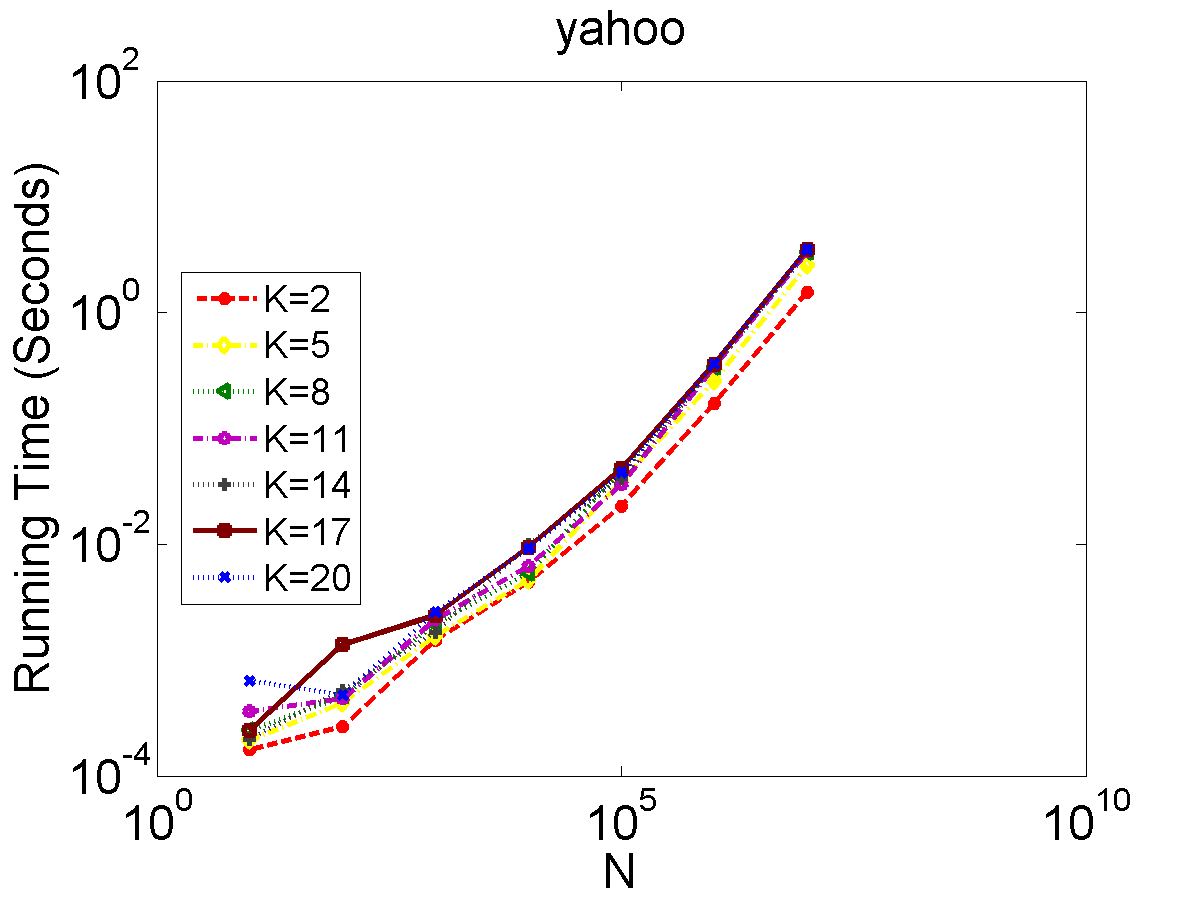}
\label{fig:rand-hk-pr-time}
}
\subfigure[rand-HK-PR conductance]{
\includegraphics[width=0.31\textwidth]{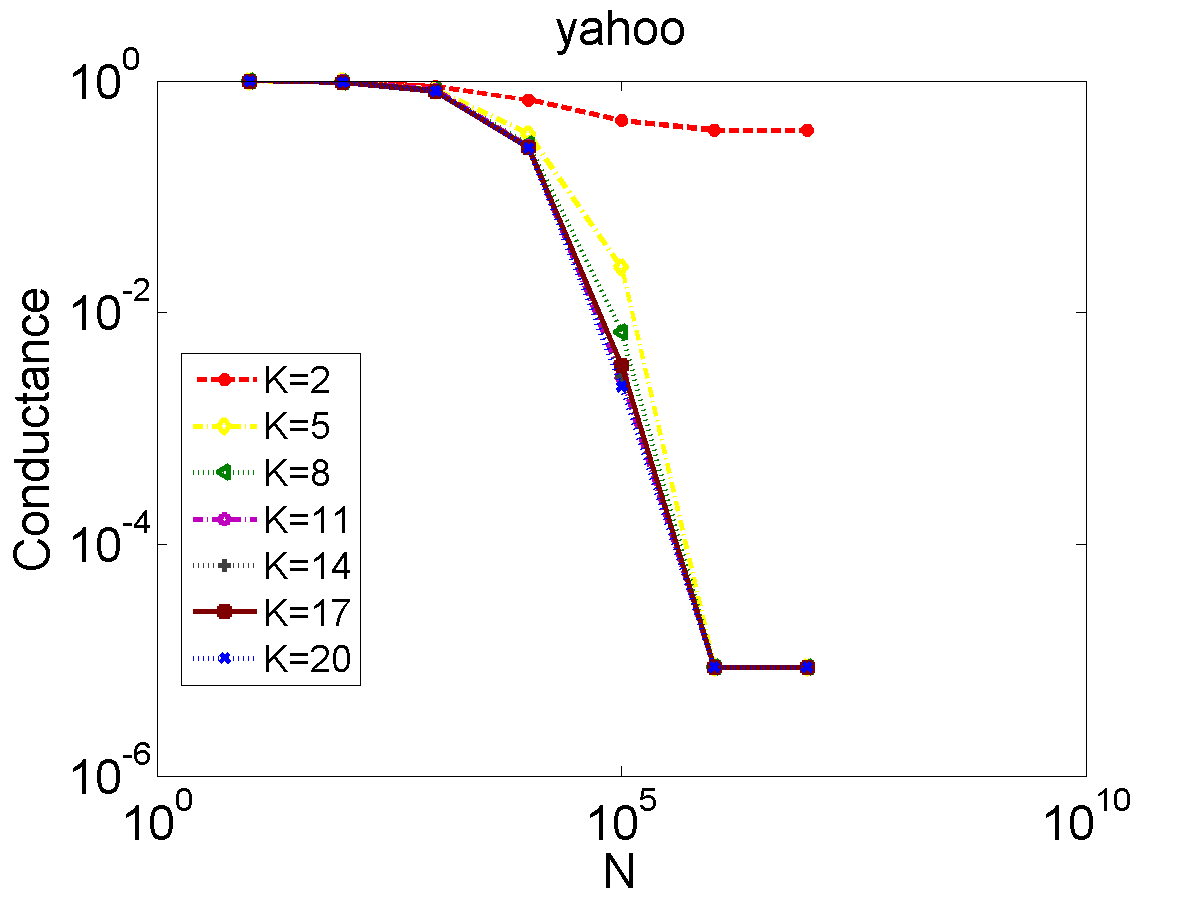}
\label{fig:rand-hk-pr-conductance}
}
\caption{Running time (seconds) and conductance of algorithms as a function of
  parameter settings on the Yahoo graph. (Best viewed in color.)}
\label{fig:parameter-plots}
\end{figure*}

\myparagraph{Parameter setting versus algorithm performance} We first
study how the setting of the various parameters in Nibble, PR-Nibble,
HK-PR, and rand-HK-PR affect their running time and the conductance of
the cluster generated. Figure~\ref{fig:parameter-plots} shows the
results of this study on the Yahoo graph, the largest graph in this
paper. All experiments start from the same seed vertex, which was
chosen by sampling $10^4$ vertices and picking the one that gave the
lowest-conductance clusters.  The trends are the same for both
sequential and parallel implementations, and the reported results are
for the sequential implementations.

As expected, for Nibble (Figures~\ref{fig:nibble-time}
and~\ref{fig:nibble-conductance}), we see that increasing $T$ and/or
decreasing $\epsilon$ leads to higher running time and improved
conductance. The same thing happens for HK-PR
(Figures~\ref{fig:hk-pr-time} and~\ref{fig:hk-pr-conductance}) when
increasing $N$ and/or decreasing $\epsilon$.  For PR-Nibble, we see
that decreasing $\epsilon$ leads to higher running time and lower
conductance (Figures~\ref{fig:pr-nibble-time}
and~\ref{fig:pr-nibble-conductance}). Finally, for rand-HK-PR, we see that
increasing $K$ and/or increasing $N$ leads to higher running time and
lower conductance (Figures~\ref{fig:rand-hk-pr-time}
and~\ref{fig:rand-hk-pr-conductance}) .

\begin{figure*}[!t]
\centering
\subfigure[Nibble]{
\includegraphics[width=0.48\textwidth]{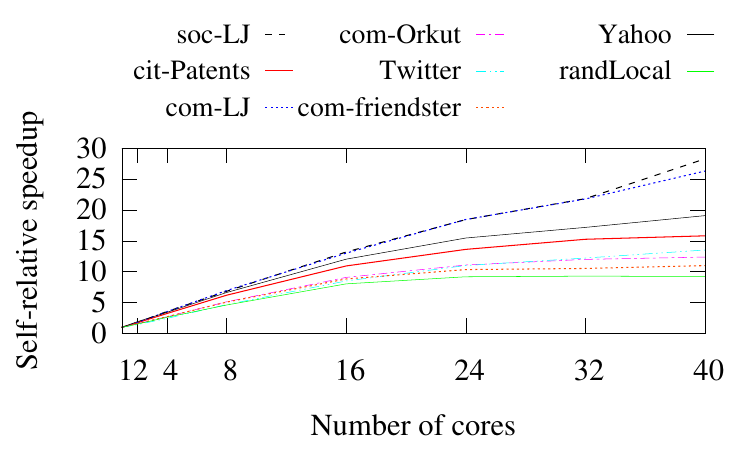}
\label{fig:bfs1}
}
\subfigure[PR-Nibble]{
\includegraphics[width=0.48\textwidth]{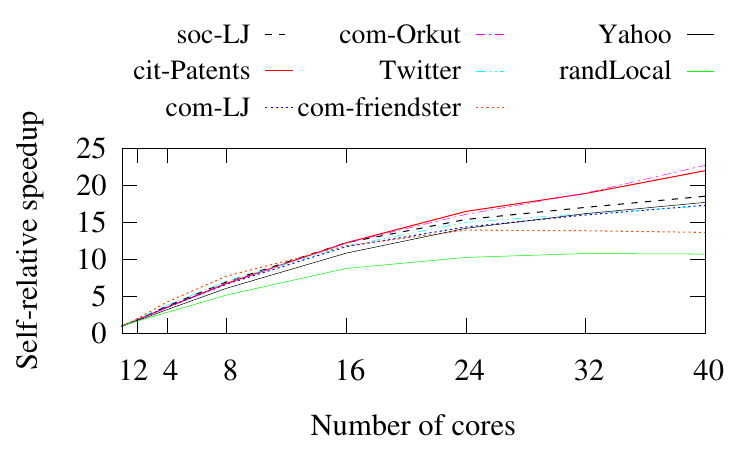}
\label{fig:pr1}
}

\subfigure[HK-PR]{
\includegraphics[width=0.48\textwidth]{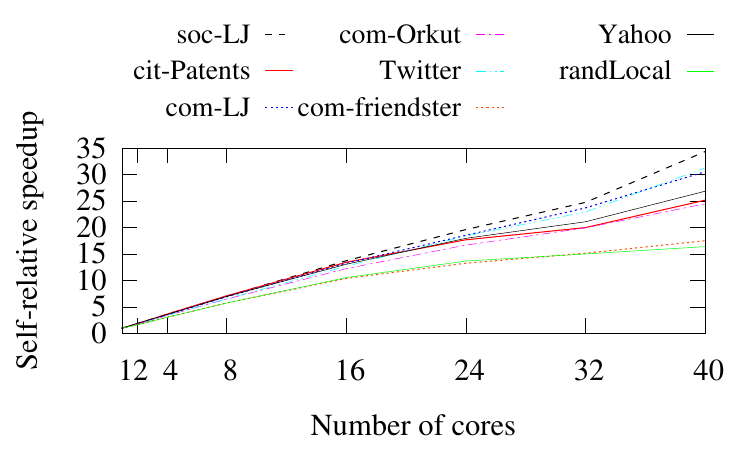}
\label{fig:speedup}
}
\subfigure[rand-HK-PR]{
\includegraphics[width=0.48\textwidth]{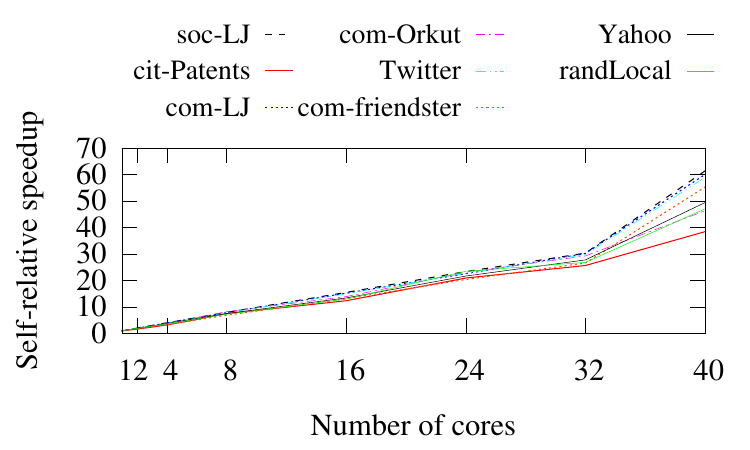}
\label{fig:speedup}
}
\caption{Parallel self-relative speedup versus cores count for the four parallel algorithms on several input graphs. On 40 cores, 80 hyper-threads are used. (Best viewed in color.)}
\label{fig:speedup-plots}
\end{figure*}

\myparagraph{Parallel Performance of Local Clustering} Here we study
the parallel scalability of our implementations of Nibble, PR-Nibble,
HK-PR, and Rand-HK-PR.  Table~\ref{table:numbers} shows the parallel ($T_{40}$)
and single-thread ($T_1$) running times of our parallel implementations, as
well as the running time of the sequential implementation,
for a setting of the parameters described in the table caption.  
All of the experiments start from a single arbitrary vertex in the
largest component.\footnote{Our codes can easily be
  modified to take as input a seed set with multiple vertices. This
  would increase the frontier sizes at each iteration, leading to more
  parallelism.}  The parameters were set so that for most graphs at
least tens of thousands of vertices were touched; otherwise the
algorithms finish in milliseconds, and there is not enough work to
benefit from parallelism.  For all graphs except for nlpkkt240 and
3D-grid, we see reasonable parallel speedup over the single-thread
times.  For nlpkkt240 and 3D-grid, not many vertices are touched as
the graphs are not well-connected, and so the experiments terminated
quickly. For these types of graphs, there are no good local clusters,
and so it may not be useful to run a local clustering algorithm. We
note that the running times of the algorithms depend highly on the
seed vertex and parameter settings, but we believe that our parallel
algorithms are useful in cases where at least tens of thousands of
vertices are touched (which is a small number for massive graphs).
For Nibble, HK-PR, and rand-HK-PR we see that the parallel version on
a single thread actually outperforms the sequential version in most
cases. For Nibble and HK-PR, we believe this is because our concurrent
hash table (used to represent the sparse sets) is more efficient than
STL's \texttt{unordered\_map}, even on one thread. For rand-HK-PR, the
parallel method uses sorting to obtain the vector rather than
maintaining it in a sparse set as in the sequential case, and this
seems to be more efficient even on a single thread in most cases.

Figure~\ref{fig:speedup-plots} shows the self-relative speedups
(relative to the algorithm's single-thread time $T_1$) of the four
algorithms versus thread count on several input graphs.  Nibble,
PR-Nibble, and HK-PR get reasonable parallel speedup (9--35x on 40
cores), although the speedup is not perfect due to memory contention
when running in parallel and also due to some frontiers
being too small to benefit from parallelism.  rand-HK-PR gets even
better speedup as most of the algorithm is embarrassingly parallel (over 40x on 40 cores
due to two-way hyper-threading).

\myparagraph{Sweep Cut Performance} Here we study the performance of
our parallel sweep cut algorithm. Table~\ref{table:numbers} shows the
running time of our parallel sweep cut implementation and the standard
sequential implementation on the output of Nibble.  The performance
trends for sweep cut were similar when applied to outputs of the other
clustering algorithms.
Except for nlpkkt240 and 3D-grid, where the input cluster was too
small to benefit from parallelism, the self-relative speedup of
parallel sweep ranges from 23 to 28 on 40 cores with hyper-threading.
On a single thread, parallel sweep is slower than sequential sweep due
to overheads of the parallel algorithm (e.g., scanning over
the edges several times instead of just once).

\begin{figure}[!t]
\centering
\includegraphics[width=0.7\columnwidth]{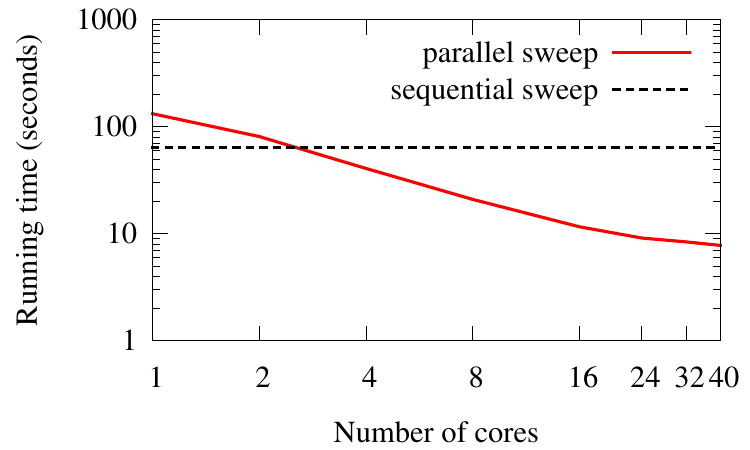}
\caption{Running time (seconds) versus core count in log-log scale of
  sweep cut on a cluster with 1.3 million vertices and 566 million
  edges. On 40 cores, 80 hyper-threads are used.}
\label{fig:sweepVsThread}
\end{figure}

Figure~\ref{fig:sweepVsThread} shows the running time of sweep cut as
a function of thread count (log-log scale). The input cluster was
generated by running Nibble on the Yahoo graph with $T=20$ and
$\epsilon=10^{-9}$. The number of vertices in the cluster is 1.3
million and its volume is 566 million.  We see that the the parallel
implementation scales well (almost linearly) with the number of
threads, and outperforms the sequential implementation with 4 or more
threads.

Figure~\ref{fig:sweep-scale} shows the running time of the parallel
sweep cut on 40 cores versus the volume of the input
set, generated by running Nibble with different parameter settings on
the Yahoo graph. We see that the running time scales nearly linearly,
which is expected since the time is dominated by linear-work
operations (the only part that scales super-linearly is the initial
sort, which takes a small fraction of the total time).

\begin{figure}[!t]
\centering
\includegraphics[width=0.7\columnwidth]{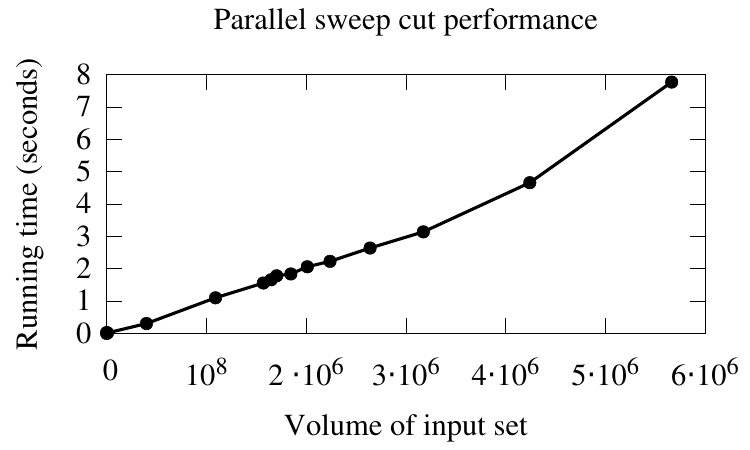}
\caption{Running time (seconds) versus cluster volume for parallel sweep cut on 40 cores with hyper-threading.}
\label{fig:sweep-scale}
\end{figure}

From Table~\ref{table:numbers}, we can see that the sweep cut
takes a significant fraction of (in some cases dominating) the overall
time of running Nibble followed by a sweep cut procedure. This was
also true for the other local clustering algorithms. Thus parallelizing
the sweep cut procedure is important in achieving good overall
performance.

\myparagraph{Network Community Profiles} 
Due to the efficiency of our parallel algorithms we are able to
quickly generate network community profile (NCP) plots for large graphs,
which show the (approximate) best conductance for clusters of a given
size in the graph versus the cluster
size~\cite{LeskovecLDM09}. Figure~\ref{fig:ncp} shows the NCPs for
several graphs with at least a billion edges, larger than any of the
graphs whose NCP has been studied
before~\cite{LeskovecLDM09,Jeub15}. The data was collected by running
PR-Nibble from $10^5$ random seed vertices and by varying $\alpha$ and
$\epsilon$.  For Twitter and com-friendster, the curves are downwards
sloping with increasing cluster size until around 10--100 vertices,
and then upwards sloping, which is consistent with the observation by
Leskovec et al.~\cite{LeskovecLDM09}, that good communities are relatively small.
For Yahoo, although there are
low-conductance clusters with small size, there also seems to be many
low-conductance clusters at larger sizes (tens of thousands of
vertices).

\begin{figure*}[!t]
\centering

\subfigure[Twitter]{
\includegraphics[width=0.31\textwidth]{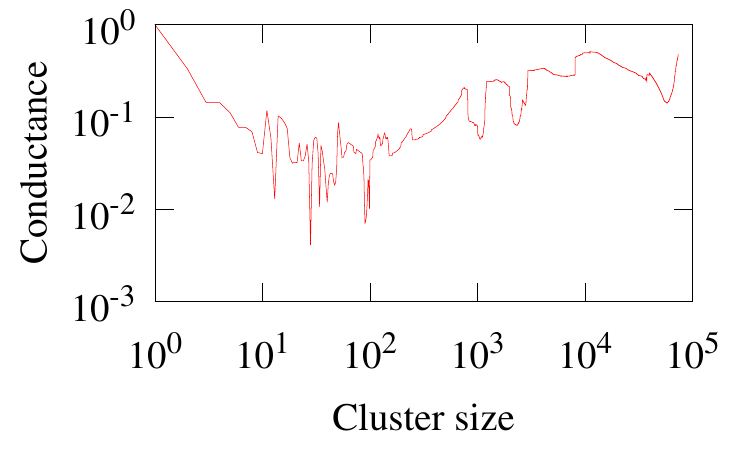}
\label{fig:bfs1}
}
\subfigure[com-friendster]{
\includegraphics[width=0.31\textwidth]{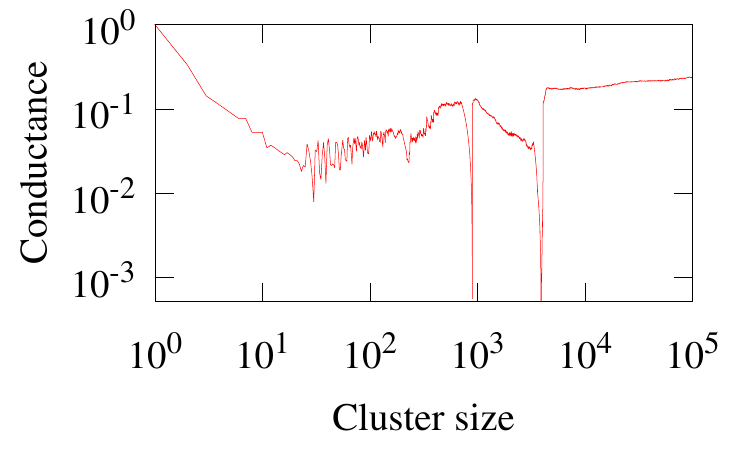}
\label{fig:pr1}
}
\subfigure[Yahoo]{
\includegraphics[width=0.31\textwidth]{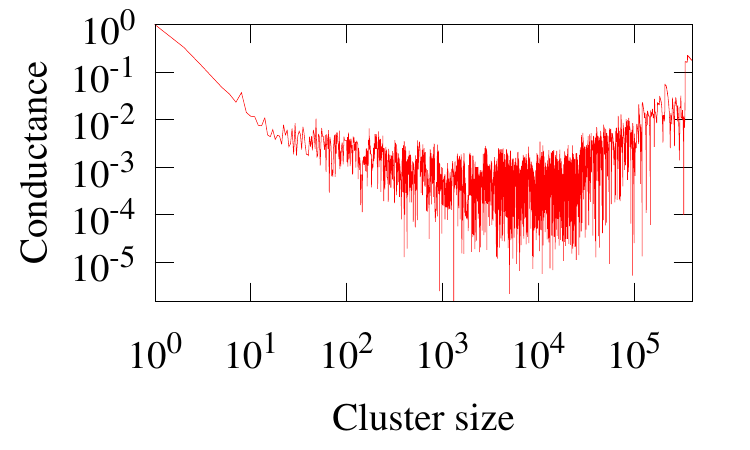}
\label{fig:speedup}
}

\caption{Network community profile (NCP) plots for billion-edge graphs.}
\label{fig:ncp}
\end{figure*}

\section{Other Related Work}\label{sec:alg}
Andersen
and Lang~\cite{AndersenL06} developed a variant of the Nibble 
algorithm with multiple starting vertices for finding communities in
graphs.
The PR-Nibble algorithm has
been extended to directed graphs~\cite{AndersenCL08}.  
A variant of PR-Nibble has also been used to obtain clusters with
better guarantees when the cluster is internally
well-connected~\cite{ZLM13}. 
Finally, statistical/optimization
perspectives on the PR-Nibble algorithm have been
proposed~\cite{GM14_ICML,Fountoulakis15}.

Andersen and Peres~\cite{Andersen2009} present an algorithm called
\emph{evolving sets}, where $f(\phi,n) = O(\sqrt{\phi\log n})$ and
that has a work bound of $O(|S|\polylog(n)/\sqrt{\phi})$. 
The algorithm maintains the position of a random walk starting at the
seed vertex. Starting with a single vertex in a set $S$, each
iteration of the algorithm adds or deletes vertices from $S$ based on
whether the probability of transitioning to a given vertex from the
current set is above some randomly chosen threshold. If in any
iteration the conductance of $S$ is at most $f(\phi,n)$, then the
algorithm returns $S$.  We implemented this algorithm but found the
behavior of the algorithm to vary widely as the random choices in each
iteration can lead to very different sets. We note that the algorithm
can be parallelized work-efficiently by using data-parallel operations
but we omit the discussion in this paper due to space constraints.  Subsequent to
the work of Andersen and Peres, there have been other local clustering
algorithms developed with stronger guarantees~\cite{Gharan2012,Kwok2012,Kwok2016}.

Mahoney et al.~\cite{MahoneyOV12} describe a spectral algorithm for
finding clusters that are locally-biased towards a seed set in terms
of satisfying an additional ``locality'' constraint in the spectral
optimization program. However, the algorithm itself is not local as it
requires work at least linear in the graph size.

Finally, there have been algorithms developed that take as
input a local cluster and return a nearby local cluster with better
conductance~\cite{Andersen2008,MahoneyOV12,OrecchiaZ14,Veldt16}.

\section{Conclusion}\label{sec:conclusion}
We have presented work-efficient parallel algorithms for local
graph clustering and our experiments demonstrated that the algorithms
achieve good performance and scalability, significantly improving the efficiency of the
exploration of local graph clusters in massive graphs.
We have performed experiments studying the output cluster conductance versus
running time of the four local algorithms but did not find any
one algorithm that always dominated the others. Since all of our
parallel algorithms are efficient, data analysts can use any of them
for graph cluster exploration, or even use all of them to find
slightly different clusters of similar size from the same seed set.

We believe that our algorithms are extremely useful in the interactive
setting where a graph is loaded once into memory and many local
cluster computations are executed on it.  In the setting where one
only wants to run a few queries, our algorithms are still very
efficient but the cost of loading the graph will not be amortized
across the queries. To improve the performance in this setting, we are
interested in developing efficient methods for traversing graphs
locally from disk. We are also interested in parallelizing local
flow-based algorithms~\cite{OrecchiaZ14,Veldt16} for improving cluster
quality.

\section*{Acknowledgements} Julian Shun is supported by the Miller
Institute for Basic Research in Science at UC
Berkeley. Farbod Roosta-Khorasani, Kimon Fountoulakis, and Michael W. Mahoney are supported by
the DARPA XDATA and GRAPHS programs.  We thank the Intel Labs Academic
Research Office for the Parallel Algorithms for Non-Numeric Computing
Program for providing the machine for our experiments. We thank Guy
Blelloch for early discussions on efficiently computing the
conductance of sets in parallel.

{
\bibliographystyle{abbrv}
\bibliography{ref}
}

\end{document}